\documentclass[11pt]{article}

\usepackage{amssymb} 
\usepackage{amsmath}     
\usepackage{amsthm}
\usepackage{todonotes}
\usepackage{mathtools}
\usepackage{a4wide}
\usepackage{graphicx}
\usepackage{hyperref}
\usepackage{caption}
\usepackage{subcaption}
\usepackage[table]{xcolor}
\usepackage{multirow}

\usepackage[ruled,vlined,linesnumbered]{algorithm2e}
\SetKw{KwNot}{not}
\SetKw{KwOR}{or}
\SetKw{KwAND}{and}
\SetKw{KwExitFor}{exit for-loop}
\SetKw{KwStep}{Step}

\usepackage{authblk}

\usepackage{natbib}

\newcommand{\Rset}{\mathbb{R}}
\newcommand{\Xset}{\mathbb{X}}

\newcommand{\Zset}{\mathbb{Z}}

\newtheorem{thm}{Theorem}

\newtheorem{cor}{Corollary}
\newtheorem{prop}{Proposition}
\newtheorem{rem}{Remark}

\newcommand*{\theadl}[1]{\multicolumn{1}{l}{\textbf{#1}}}

\begin{document}

\title{Single-item  lot sizing problem under budgeted lead-time uncertainty}

\author[1]{Romain Guillaume}
\author[2]{Adam Kasperski\footnote{Corresponding author}}
\author[2]{Szymon Wróbel}
\author[2]{Pawe{\l} Zieli\'nski}

\affil[1]{Universit{\'e} de Toulouse-IRIT, Toulouse, France,
\texttt{romain.guillaume@irit.fr}}

\affil[2]{
Wroc{\l}aw  University of Science and Technology, Wroc{\l}aw, Poland\\
            \texttt{\{adam.kasperski, szymon.wrobel, pawel.zielinski\}@pwr.edu.pl}}

\date{}

\maketitle

 \begin{abstract}
In this paper, a single-item  lot sizing problem with backordering is discussed. 
The time horizon is divided into planning periods, characterized by 
fixed and variable production costs, and future delivery periods with 
specified demands, where inventory holding and backordering costs 
may occur. For each planning period, a common nominal lead time is given.
The true lead times can deviate to some extent from the nominal one, and their exact values are unknown at the planning step.  
We assume that lead times take only integer values and splitting production orders is not allowed. Furthermore, order crossovers are prohibited; thus, an order placed earlier cannot arrive after one placed later.
 A budgeted uncertainty set of possible lead-time scenarios is defined, where a budget allows us to control the amount of uncertainty of 
lead times. It is shown how to construct a family of production plans varying from the most optimistic (a best lead-time scenario occurs) to the most pessimistic (a worst lead-time scenario occurs). In order to compute these plans the \textsc{R*} criterion is applied which generalizes the conservative robust min-max criterion, commonly used in robust optimization. The computational complexity of the problem is investigated. Polynomial, pseudopolynomial time algorithms, and  mixed integer programming formulations are proposed to solve the general problem and its special cases.  The results of computational tests are provided that demonstrate that using the \textsc{R*} criterion can significantly enlarge the set of candidate production plans.

  \end{abstract}
  
\noindent \textbf{Keywords}: dynamic lot sizing, lead-time uncertainty, robust optimization, production planning

\section{Introduction}

The single-item  lot-sizing problem is a basic production and distribution planning model in which there is a demand for a single product over a planning horizon. This demand should be satisfied by placing production orders in previous periods. The time from placing an order to receiving the delivery is called a lead time. 
By incorporating lead times, we obtain a delivery plan that may incur inventory holding and backorder costs.
The basic single-item  lot sizing problem was introduced by~\cite{WW58} and a description of its different variants can be found in the survey by~\cite{BADN17} or in~\cite{WP06}. In this paper, we consider a version with backordering in which the demand can be  temporarily unmet and later satisfied.   

 In the real world, the  lead times are often uncertain and may not be fully known at the time of planning. 
 A simple model of uncertainty assumes that the lead times can take their values within some ranges. To control the amount of uncertainty, a budget can be introduced that limits the number of lead times that differ from the nominal ones (or limits the total deviation from the nominal lead times). The concept of budgeted uncertainty was introduced by~\cite{BS04} and has been widely adopted in robust optimization. Applying the robust optimization framework~(see, e.g.~\cite{BN09}, \cite{KY97}), we can 
identify production plans that minimize the total production, inventory 
holding, and backordering costs under a worst-case lead-time scenario.
 The application of the robust min-max approach to lead-time uncertainty in lot-sizing models has already been explored in the literature. In~\cite{TA17} a single-item lot-sizing problem with uncertain lead times and uncertain demands was discussed. 
 A budgeted uncertainty set for lead times was constructed. A relaxed version of this set was considered, which means that partial orders can be delivered in subsequent periods.  
 Taking into account both demand and lead-time uncertainty, the robust min-max problem was solved by a Bender's decomposition algorithm. In~\cite{TON22} a variant of the problem with many suppliers was investigated. Again, a budgeted uncertainty set for uncertain lead times was constructed. The authors in~\cite{TON22} introduced three budgets that limit the maximum number of late orders, the number of late deliveries, and the total lateness. 
 However, as in~\cite{TA17}, a relaxed version was considered in which an order can be split and arrive at more than one future period. Several solution algorithms were proposed and tested to solve the problem. Yet another robust approach to lead-time uncertainty was proposed by~\cite{HA17}, where the uncertain lead times were modeled by providing a discrete uncertainty set containing possible lead-time realizations (scenarios). 
 
 In this paper, we divide the time horizon into planning periods, where 
fixed and variable production costs are incurred, and future delivery 
periods with specified demands. We assume a common nominal lead time 
across all planning periods.
This assumption can be satisfied, for example, for a short planning horizon, when the orders come from one supplier.
We consider a simple uncertainty model in which a minimum and a maximum deviation from the nominal lead time are specified for each planning period. Hence, a delivery can be both late or early.
The true lead time can take discrete values within a given range. 
In contrast to the models discussed by~\cite{TA17} and  \cite{TON22}, order 
splitting is not allowed; thus, each order must arrive in its entirety 
at a single future period. Furthermore, we assume that order crossovers 
are prohibited, meaning that an order placed earlier cannot arrive after the one 
placed later. We show that this realistic assumption is crucial and leads to efficient algorithms or more compact mixed integer programming reformulations for the problem.
We define a restricted lead-time uncertainty set that satisfies the 
aforementioned assumptions by introducing two types of budget. 
The first budget constrains the number of lead times that deviate from 
their nominal values, while the second limits the total deviation 
of all lead times from the nominal ones.

The robust min-max criterion is often regarded as too conservative. 
In practice, it is often beneficial to provide decision makers with 
a broader spectrum of solutions, ranging from optimistic to pessimistic.
 In~\cite{FG20} the so-called \textsc{R*} criterion was proposed. The idea is to introduce a cost budget $B$ that represents a maximum amount that the decision-maker is willing to pay to implement a solution. We then seek a solution that minimizes the minimum cost in the lead-time scenarios, subject to the condition that its maximum total cost does not exceed $B$. For small values of $B$ we get the traditional min-max robust problem. Choosing larger values of $B$ we obtain more optimistic solutions that take into account good lead-time scenarios.  The ultimate choice depends on some outer factors, such as decision-maker risk aversion. 
This paper is a significantly expanded version of the extended abstract by~\cite{GKZ24}.

This paper contains the following new results. We show that for a given production plan, a best and a worst lead-time scenario can be found in polynomial time. In the absence of setup costs, the robust min-max problem can be solved in polynomial time using a linear programming reformulation (for arbitrary setup costs, a mixed integer programming model exists). On the other hand, a best (optimistic) production plan can be found in pseudopolynomial time for arbitrary setup costs. 
We prove that the problem with the \textsc{R*}~criterion is NP-hard even in the 
special case without backordering and with equal production capacities. 
Note that the deterministic version of this problem is solvable in 
polynomial time (see~\cite{FLK80}). 
To solve the problem with the  \textsc{R*}~criterion, a mixed integer programming formulation is constructed.

The paper is organized as follows. In Section~\ref{sec:problemform} we recall the formulation of the deterministic single-item lot sizing problem. Then, we introduce two budgeted uncertainty sets for lead times and formulate several optimization problems for computing the pessimistic, optimistic and  \textsc{R*} production plans. We also show that a special case of the problem with the  \textsc{R*}~criterion is NP-hard. In Section~\ref{sec:solving}, we propose polynomial and pseudopolynomial 
time algorithms, as well as mixed-integer programming formulations, to 
solve the general \textsc{R*}~problem and its special cases. We show 
that computing the best and worst lead-time scenarios can be performed 
in polynomial time by solving a variant of the constrained shortest 
(longest) path problem. Furthermore, we develop methods for computing 
pessimistic, optimistic, and \textsc{R*}~production plans. 
Section~\ref{sec:tests} presents the results of computational 
experiments, where the obtained production plans are evaluated using a
Monte Carlo simulation. We demonstrate how the choice of problem 
parameters influences the distribution of solution costs. We also show that the \textsc{R*} criterion can significantly enlarge the set of candidate production plans.

\section{Problem formulation}
\label{sec:problemform}

In this section, we first recall a formulation of the deterministic single-item lot sizing problem. Then, we introduce an uncertainty model for lead times and formulate several optimization problems for computing the pessimistic, optimistic and  \textsc{R*}production plans.
We also establish NP-hardness of the \textsc{R*}~problem and
its inapproximability, if the assumption on order crossovers is relaxed.

\subsection{Deterministic single-item   lot-sizing problem}
\label{secdet}

We are given a set $[T]=\{1,\dots, T\}$  of \emph{planning periods} and a set $[T^+]=\{T+1,\dots, T^+\}$ of \emph{future periods}.  
For each future period $t\in [T^{+}]$, there is a \emph{demand}~$d_t\geq 0$. We are  given \emph{setup} and
 \emph{production} costs $c^S_t$, $c^P_t$ for each planning period $t\in [T]$, and \emph{inventory holding} and \emph{backorder} costs $c^I_t$, $c^B_t$ for each future period $t\in [T^+]$.  
  Let $\pmb{x}=(x_t)_{t\in [T]}$ be a \emph{production plan}, where $x_t\geq 0$ is the amount of production in the planning period $t\in [T]$. Let us denote by $\Xset$ the set of \emph{production plans}. In this paper, we assume that 
$$
\Xset=\{\pmb{x}\in \Rset^{T}_{+}: \;   0\leq x_{t}\leq C_{t}, \; t\in [T]\},                               
$$
where $C_t$ is the \emph{capacity} in period $t\in [T]$, i.e, the maximum amount of production in period~$t$.

The production plan $\pmb{x}=(x_t)_{t\in [T]}$ generates a \emph{procurement (delivery) plan} in the future periods by applying \emph{lead times} $(l_t)_{t\in [T]}$ to $\pmb{x}$, that is production $x_t$ in period $t\in [T]$ arrives at future period $t+l_t\in [T^+]$.
The lead times $l_t\in \Zset_{+}$, $t\in [T]$, are positive integers, so splitting of production orders is not allowed.
We assume that the production in period $t\in [T]$ cannot arrive after the production in the next period $(t+1)\in [T]$, that is, $t+l_t\leq (t+1)+l_{t+1}$. This assumption, called  \emph{no order crossover}, can be expressed as:
\begin{equation}
\label{assump}
l_t\leq l_{t+1}+1,\; t\in [T-1].
\end{equation}

The deterministic  single-item lot-sizing problem with lead times and backordering
can be modeled by the following mixed integer program, called \emph{aggregate formulation} (see, e.g.,~\cite{BADN17}):
\begin{align}
\min &
	\sum_{t\in [T]} (c_t^S y_t+c_t^P x_t)+
	\sum_{t\in [T^+]} (c^{I}_t I_{t}+ c^B_t B_{t})&\label{spp1}\\ 
\text{s.t.}
	& B_{t}- I_{t}=D_{t}-O_t & t \in [T^+] \label{spp2}\\
	& O_t=\sum_{i\in [T]: i+l_i\leq t} x_{i} & t\in [T^+] \label{spp3}\\
	& x_{t}\leq C_ty_t & t\in [T] \label{spp4}\\
	& B_{t},I_{t}\geq 0  & t\in [T^+] \label{spp5}\\
	& y_{t}\in \{0,1\}  & t\in [T] \label{spp6}\\
	& \pmb{x}\in \Xset\subseteq \Rset^{T}_{+}, \label{spp7}
\end{align}
where $D_{t}=\sum_{i\in [T^+]}d_{i}$ is
  the \emph{cumulative demand} up to period~$t\in [T^+]$, $I_{t}$, $B_{t}$  and $O_{t}$ are  the \emph{inventory level},  \emph{backorder level} and \emph{cumulative delivery}  in period~$t\in [T^+]$, respectively.
Equation~(\ref{spp3}) computes the cumulative delivery $O_t$ in $t\in [T^+]$ from the lead times and the production amounts in $[T]$. The constraints~(\ref{spp4}) express the fact that a non-zero production in period $t\in [T]$ incurs  the setup cost~$c_t^S$ in $t$.

\subsection{Uncertain lead times}
\label{sec:uncelead}

Assume that the true values of lead times are unknown at the planning step, when the production plan must be established. This uncertainty is modeled by a \emph{lead-time uncertainty set}, which contains possible realizations of lead times.
Let $\widehat{l}\in \Zset_{+}$,
$\widehat{l}\geq 1$, be a common \emph{nominal lead time}, which is the same for each planning period $t\in [T]$. For each period $t\in [T]$,  deviations $\zeta^+_t$, $\zeta^-_t$ from $\widehat{l}$ are specified. The values of  $\zeta^+_t$, $\zeta^-_t$ are nonnegative integers.
 Hence $l_t \in  \{l^{\min}_t,l^{\min}_t+1,\dots, l^{\max}_t\}$, $t \in [T]$, with $l^{\min}_t=\widehat{l}-\zeta^-_t$ and $ l^{\max}_t=\widehat{l}+\zeta^+_t $. 
 We assume that $t+l^{\min}_t>T$ and $t+l^{\max}_t\leq T^+$ for each $t\in [T]$, so production in period $t\in [T]$ must arrive at some future period in $[T^+]$.
 From now on, we assume that $l^{\min}_t\leq l^{\min}_{t+1}+1$ and $l^{\max}_t\leq l^{\max}_{t+1}+1$ for each $t\in [T-1]$, which follows from the no order crossover assumption~(\ref{assump}) about the lead times made in Section~\ref{secdet}. 
We will use $\pmb{\ell}=(l_1,\dots,l_t)$ to denote a \emph{lead-time scenario}, i.e. a particular realization of the lead times. 
 
 Let $\mathcal{L}\subset  \Zset_{+}^{T}$ be the \emph{lead-time uncertainty set}, where each scenario 
 $\pmb{\ell}=(l_1,\dots,l_t)\in \mathcal{L}$ satisfies 
 the  constraints:
 \begin{align}
 &  l_t \in \{l^{\min}_t,l^{\min}_t+1,\dots, l^{\max}_t\} & t\in [T],\label{clun1}\\
  & l_t\leq l_{t+1}+1 & t\in [T-1],\label{clun2}\\
  & l_t\in  \Zset_{+} & t\in [T]. \label{clun3}
\end{align}
In the following, we will consider two \emph{budgeted versions} of the uncertainty set $\mathcal{L}$:
\begin{align}
& \mathcal{L}(\Gamma^d)=\{\pmb{\ell}\in \mathcal{L}: |\{t\in [T]: l_t\neq \widehat{l}\}|\leq \Gamma^d\}, \label{dus}\\
& \mathcal{L}(\Gamma^c)=\{\pmb{\ell}\in \mathcal{L}: \sum_{t\in [T]} |l_t-\widehat{l}|\leq \Gamma^c\}. \label{cus}
\end{align}
In the uncertainty set $\mathcal{L}(\Gamma^d)$, the \emph{budget} $\Gamma^d\in\{0,\dots,T\}$ limits the number of lead times that are different from the nominal lead time $\widehat{l}$. This type of budgeted uncertainty was introduced by~\cite{BS04}. 
On the other hand, the \emph{budget} $\Gamma^c\in \Zset_{+}$ limits the total deviation of the lead times from $\widehat{l}$. Setting $\Gamma^d=\Gamma^c=0$, we get the nominal problem in which all lead times are equal to~$\widehat{l}$. After choosing $\Gamma^d\geq T$ or $\Gamma^c\geq T(T^+-T)$, the uncertainty sets $\mathcal{L}(\Gamma^c)$ and $\mathcal{L}(\Gamma^d)$ reduce to $\mathcal{L}$.

\subsection{Single-item lot sizing problem with uncertain lead times}
\label{sec:unclead}
 
Let $\mathcal{C}(\pmb{x},\pmb{\ell})$ be the cost of the production plan $\pmb{x}\in \mathbb{X}$ under the lead-time scenario $\pmb{\ell}\in \mathcal{L}(\Gamma)$, where $\Gamma\in \{\Gamma^d,\Gamma^c\}$. Hence, $C(\pmb{x},\pmb{\ell})$ is the optimal objective value of (\ref{spp1})-(\ref{spp7}) for a fixed production plan $\pmb{x}=(x_1,\dots,x_t)\in \mathbb{X}$ and lead-time scenario $\pmb{\ell}=(l_1,\dots,l_t)\in \mathcal{L}(\Gamma)$. 
We will compute an \emph{optimistic production plan} $\pmb{x}^{o}$ by solving the problem
\begin{equation}
\label{probopt}
\textsc{Min-Min}:\; \pmb{x}^{o}\in \arg \min_{\pmb{x}\in \Xset}\min_{\pmb{\ell}\in \mathcal{L}(\Gamma)} \mathcal{C}(\pmb{x},\pmb{\ell}).
\end{equation}
We will also compute a \emph{pessimistic (robust) production plan} $\pmb{x}^{p}$ by solving the traditional robust min-max problem
\begin{equation}
\label{probpes}
\textsc{Min-Max}:\;  \pmb{x}^{p}\in \arg \min_{\pmb{x}\in \Xset}\max_{\pmb{\ell}\in \mathcal{L}(\Gamma)} \mathcal{C}(\pmb{x},\pmb{\ell}).
\end{equation}
We will also be interested in solving the inner problems in~(\ref{probopt}) and~(\ref{probpes}), that is, computing a best (optimistic) and worst (pessimistic) lead-time scenario for a given production plan $\pmb{x}$:
\begin{equation}
\label{bleadprob}
 \pmb{\ell}^{p}\in \arg \max_{\pmb{\ell}\in \mathcal{L}(\Gamma)} \mathcal{C}(\pmb{x},\pmb{\ell}),
 \end{equation}
 \begin{equation}
 \label{wleadprob}
 \pmb{\ell}^{o}\in \arg \min_{\pmb{\ell}\in \mathcal{L}(\Gamma)} \mathcal{C}(\pmb{x},\pmb{\ell}).
\end{equation}
Observe that the interval $[\mathcal{C}(\pmb{x},\pmb{\ell}^o), \mathcal{C}(\pmb{x},\pmb{\ell}^p)]$ contains all possible values of the cost of $\pmb{x}$ and thus characterizes the uncertainty associated with $\pmb{x}$.

Production plans $\pmb{x}^o$ and $\pmb{x}^p$ represent two extreme attitudes towards a risk. To compute alternative production plans that represent a tradeoff between these two cases, another criterion should be applied. 
One possible choice is the well-known Hurwicz criterion (see, e.g.,~\cite{LR57}), which is defined as a convex combination of the minimum and maximum costs of $\pmb{x}$ over the set $\mathcal{L}(\Gamma)$. 
However, this approach has been criticized in the context of sequential decision-making (see, e.g.,~\cite{DFGT15}, \cite{FG20}). 
Indeed, production planning can be modeled as a sequential decision problem under uncertainty and represented as a decision tree
with production in periods as decision nodes and lead time scenarios as chance nodes. 
A solution is thus viewed as a strategy that assigns a decision (production amount) to each decision node.
Within this framework, the Hurwicz criterion fails to satisfy dynamic consistency; that is, the computed strategy can be nonoptimal 
 for a subset of the decision tree restricted to some node and all its descendants.
  Moreover, as highlighted by~\cite{FG20} and \cite{krug2020decision}, 
 the Hurwicz criterion suffers from another serious drawback, namely, it relies solely on a convex combination of the best and worst outcomes, it can only identify a limited subset of Pareto-optimal solutions.

In this paper we apply an alternative criterion, proposed by~\cite{FG20}, called \textsc{R*},
which preserves the essential properties of sequential decision processes while significantly expanding the set of reasonable candidate solutions.
For a description of its properties, in particular for its axiomatic characterization, we refer the reader to~\cite{FG20}. Let $B\in \Rset_{+}$ be a quantity, called a \emph{cost budget}, that represents the maximum amount the decision-maker is willing to pay for a production plan. We consider the following optimization problem:
\begin{equation}
\label{probr}
\text{\textsc{R*}}:\; 
	\begin{array}{lrl}
 		  \pmb{x}^r\in&  \displaystyle\arg \min_{\pmb{x}\in\Xset} & \displaystyle \min_{\pmb{\ell}\in \mathcal{L}(\Gamma_1)} \mathcal{C}(\pmb{x},\pmb{\ell}) \\
		   & \mathrm{s.t.} & \displaystyle \max_{\pmb{\ell}\in \mathcal{L}(\Gamma_2)} \mathcal{C}(\pmb{x},\pmb{\ell})\leq B,	
	\end{array}
\end{equation} 
where $\Gamma_1\in\{\Gamma^d,\Gamma^c\}$, $\Gamma_2\in\{\Gamma^d,\Gamma^c\}$.
We call an optimal solution~$ \pmb{x}^r$ to \eqref{probr} the \textsc{R*} production plan.
If the cost-budget constraint is not satisfied, i.e.
 $\max_{\pmb{\ell}\in \mathcal{L}(\Gamma_2)} \mathcal{C}(\pmb{x},\pmb{\ell})> B$ for each $\pmb{x}\in\Xset$, then we choose the pessimistic plan $\pmb{x}^p$ as the optimal solution. The problem~(\ref{probr}) generalizes~(\ref{probopt}) and~(\ref{probpes}). 
Indeed, we get~(\ref{probopt}) by choosing very large $B$ and~(\ref{probpes}) by choosing $B=0$. In~(\ref{probr}), we choose an optimistic  solution in a limited set of production plans, where the plans with high worst-case costs are excluded.
The values of the budgets $\Gamma_1$ and $\Gamma_2$ can be different. In one natural case, we can set $\Gamma_1=0$, so we solve the nominal problem with the cost-budget constraint. However, decision-makers may seek for more opportunities by choosing a positive value for~$\Gamma_1$.

The $\text{\textsc{R*}}$ problem can be viewed as a bi-criterion optimization problem. The first criterion, which accounts for optimistic lead-time scenarios, is optimized, while the second, focusing on pessimistic lead-time scenarios, is constrained by a given threshold $B$.
 By considering various combinations of $\Gamma_1$, $\Gamma_2$ and $B$, we obtain a family of production plans that can be offered to decision makers. In Section~\ref{sec:tests} we show how to evaluate these production plans. We also demonstrate
  that the \textsc{R*} criterion can lead to better production plans than the traditional robust min-max approach.

It is worth pointing out
that for the deterministic problem (\ref{spp1})-(\ref{spp7}), there is always an integral optimal production plan, provided that all demands and capacities are integral. Indeed,
if we fix the optimal values for binary variables $y_t$, $t\in [T]$,  then an optimal production plan can be computed by solving 
the corresponding 
 min-cost flow problem (see, e.g.,~\cite{AMO93}) that always has an integral optimal solution in this case. 
 In Section~\ref{secoptymplan}, we show that the same property holds for problem~(\ref{probopt}), so there is always an integral optimistic production plan.  However, while computing a pessimistic or \textsc{R*} production plan by solving~(\ref{probpes}) or~(\ref{probr}), respectively, it can be profitable to use fractional solutions, which is demonstrated by a sample instance shown in Figure~\ref{fig:fract}.
\begin{figure}[h!t]
	\centering
	\includegraphics[scale=0.9]{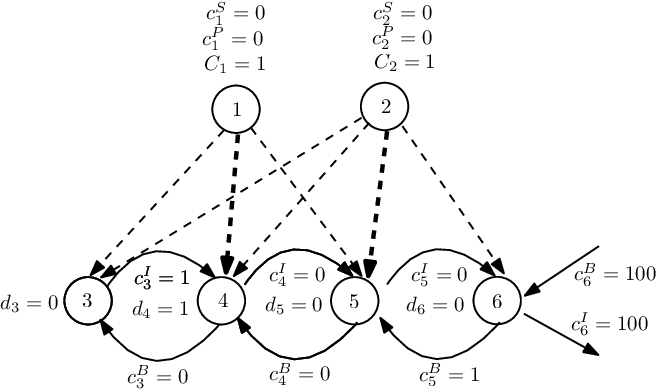}
	\caption{An instance with $T=2$, $T^+=6$, $\hat{l}=\pmb{3}$,  $\zeta^+_1=\zeta^-_1=1$, 
		$\zeta^+_2=1$, $\zeta^-_2=2$,
	$l_1\in\{2,\pmb{3},4\}$, $l_2 \in\{1,2,\pmb{3},4\}$, $\Gamma^d=1$.}
	\label{fig:fract}
\end{figure}

There are four integral production plans: $\pmb{x}^{(1)}=(0,0)$, $\pmb{x}^{(2)}=(0,1)$, $\pmb{x}^{(3)}=(1,0)$, $\pmb{x}^{(4)}=(1,1)$. The plans $\pmb{x}^{(1)}$ and $\pmb{x}^{(4)}$ are very expensive due to the large inventory holding and backorder costs in period~6. For the production plan $\pmb{x}^{(3)}$ the worst lead-time scenario is $(2,3)$ which yields the cost~1 (inventory cost in period~3 occurs). For the production plan $\pmb{x}^{(2)}$ the worst lead-time scenario is $\pmb{\ell}=(3,4)$ which also yields the cost~1 (backorder cost in period 5 occurs). However, for the production plan $(0.5, 0.5)$, the cost under any lead-time scenario is at most $0.5$ (notice that $l_1=3$ or $l_2=3$, because $\Gamma^d=1$). Therefore, it is profitable to diversify the production among the periods to obtain a more robust plan. This fact also means that the assumption about the integrality of production plans (often reasonable in practice) is essential in problems~(\ref{probpes}) and~(\ref{probr}).
\begin{rem}
\label{remint}
Problems~(\ref{probpes}) and~(\ref{probr}) do not have the integrality property.
\end{rem}

\subsection{The computational complexity of the \textsc{R*} problem}

The deterministic lot-sizing  problem (\ref{spp1})-(\ref{spp7}) is NP-hard (see~\cite{florian1971deterministic}, \cite{FLK80}).  In what follows, the problems~(\ref{probopt}), (\ref{probpes}) and (\ref{probr}) are also NP-hard, because we get the deterministic problem by fixing 
$\Gamma^{d}=0$  or $\Gamma^{c}=0$ and sufficiently large the cost budget~$B$. 
However, the version of the deterministic problem  (\ref{spp1})-(\ref{spp7}) without backordering ($c^B_t=+\infty$ for each $t\in [T^+]$) and equal capacities ($C_t=c$ for every $t\in [T]$) is polynomially solvable (\cite{florian1971deterministic}, \cite{L73}). The following theorem demonstrates that the corresponding problem with the \textsc{R*} criterion becomes NP-hard, even for 
the lead-time uncertainty set~$\mathcal{L}$ (see~(\ref{clun1})--(\ref{clun3})).
 \begin{thm}
 \label{thmcompl}
Problem~(\ref{probr}) (the \textsc{R*}~problem) without backordering and equal capacities  is NP-hard, under the lead-time uncertainty set~$\mathcal{L}$.
 Furthermore, if the integrality assumption is imposed on production plans, then the problem remains NP-hard even if all setup costs are equal to~0.
  \end{thm}
 \begin{proof}
See Appendix~\ref{dod}.
\end{proof}
From Theorem~\ref{thmcompl}, we  immediately get the following corollary.
\begin{cor}
\label{corcompl}
Problem~(\ref{probr}) (the \textsc{R*}~problem) without backordering and equal capacities   is NP-hard, under 
the budgeted uncertainty sets $\mathcal{L}(\Gamma^d)$ and $\mathcal{L}(\Gamma^c)$.
 Furthermore, if the integrality assumption is imposed on production plans, then the problem remains NP-hard even if all setup costs are equal to~0.
\end{cor}
Corollary~\ref{corcompl} results from the fact that 
the budgeted uncertainty sets $\mathcal{L}(\Gamma^d)$ and $\mathcal{L}(\Gamma^c)$
 sets reduce to $\mathcal{L}$ for sufficiently large budgets $\Gamma^c$ and $\Gamma^d$ (see Section~\ref{sec:uncelead}). It is worth pointing out that the above hardness results hold if the order 
crossovers are not allowed (see~(\ref{clun2})). However, if we drop this 
assumption in the definition of the lead-time uncertainty set~$\mathcal{L}$ 
(see~(\ref{clun1})--(\ref{clun3})), then the \textsc{R*}~problem becomes 
significantly harder. This fact follows from the next theorem.
 \begin{thm}
 \label{thmminmin}
The problem of computing a best (optimistic) lead-time scenario~$\pmb{\ell}^{o}$ for a  given production plan $\pmb{x}$
(see the problem~(\ref{wleadprob})) with order crossovers, without backordering and without setup costs
is NP-hard and inapproximable, 
unless $\text{P} \neq \text{NP}$.
  \end{thm}
 \begin{proof}
See Appendix~\ref{dod}.
\end{proof}
Since the problem of computing an optimistic lead-time scenario for a 
given production plan is a subproblem of both the \textsc{Min-Min} and 
\textsc{R*} problems, Theorem~\ref{thmminmin} leads to the following corollary.
\begin{cor}
\label{corcapx}
The \textsc{Min-Min} and \textsc{R*} problems with order crossovers, 
without backordering and without setup costs are NP-hard and inapproximable, unless $\text{P} \neq \text{NP}$.
\end{cor}
In the next section, we show that with the no order crossover assumption, the optimistic lead-time scenario can be computed in polynomial time. Therefore, this assumption is crucial when designing efficient algorithms for the problem.

\section{Solving the \textsc{R*}~problem}
\label{sec:solving}

In this section, we develop solution methods for the \textsc{R*}~problem 
and its special cases. We show that computing the best and worst 
lead-time scenarios for a given production plan can be done in 
polynomial time. Furthermore, we provide algorithms for determining 
pessimistic, optimistic, and \textsc{R*}~production plans.

\subsection{Computing the best and worst lead-time scenarios for a given production plan}
\label{secwblead}

We now show how to compute  a worst (pessimistic) and a best (optimistic) lead-time scenarios $\pmb{\ell}^o$ and $\pmb{\ell}^p$ for a given production plan $\pmb{x}\in \mathbb{X}$ by solving adversarial problems~(\ref{bleadprob}) and~(\ref{wleadprob}).
We show that both scenarios can be found by solving a \emph{constrained shortest} and a \emph{constrained longest path} problem, respectively, in a particular acyclic digraph.

 Recall that in the \emph{constrained shortest path problem} (\textsc{CSP} for short), we are given a directed graph $G=(V,A)$ with costs $c_e$ and integer transition times $t_e$ specified for each arc $e\in A$,
 with the source~$s\in V $ and the  sink~$t\in V$. We seek an $s$-$t$ path in $G$ with the smallest total cost, subject to the condition that its total transition time does not exceed a given number $K$. The \textsc{CSP} problem can be solved in $O(|A| K)$ time by the algorithm proposed by~\cite{H92}. The \emph{constrained longest path problem} (\textsc{CLP}) in acyclic digraph is defined in a similar way and has the same computational complexity.

We start with a characterization of all possible cumulative deliveries in a given period $t\in [T^+]$.  Let us define
 $I_t=\max \{i\in [T]:\; i+l^{\max}_i\leq t\}$ and $J_t=\max\{i\in [T]:\;i+l^{\min}_i\leq t\}$.  
 Observe that production in periods $1, \dots, I_t$ must arrive by period $t$, and production in periods $1, \dots, J_t$ may arrive by period $t$. We set $I_t = 0$ if $i + l_i^{\max} > t$ for each $i \in [T]$, and $J_t = 0$ if $i + l_i^{\min} > t$ for each $i \in [T]$. Note that $0 \leq I_t \leq J_t \leq T$. 
 A proof of the following result is immediate.
\begin{prop}
\label{prop_delta}
	For each period $t\in [T^+]$
   \begin{equation}
   \label{edelta}
   O_t\in \mathcal{O}_t=\left\{ \sum_{i=1}^{I_t}x_i, \sum_{i=1}^{I_t+1}x_i,\dots , \sum_{i=1}^{J_t}x_i \right\},
   \end{equation}
   where we set $\sum_{i=1}^{0} x_i:=0$.
\end{prop}
 \begin{figure}
	\centering
	\includegraphics[scale=0.9]{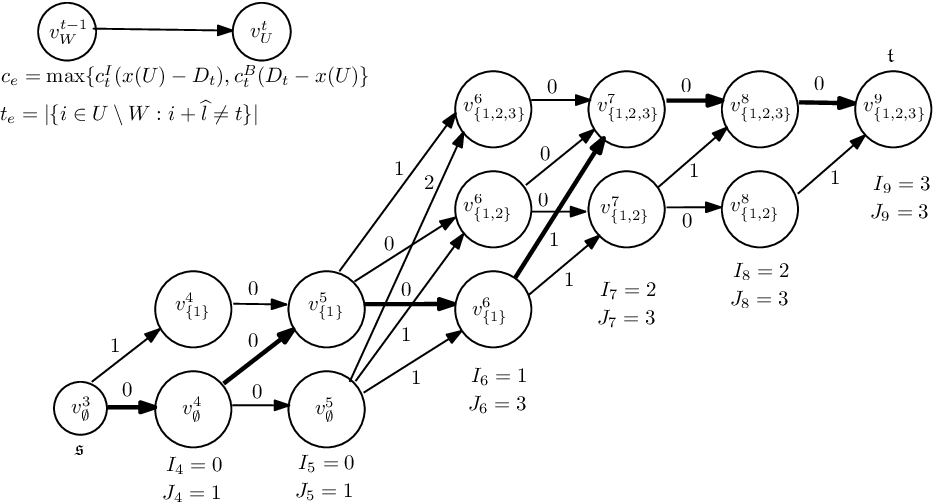}
	\caption{A sample network for $T=3$, $T^+=9$, $\pmb{x}=(x_1,x_2,x_3)\in \Xset$, $\widehat{l}=\pmb{4}$,  $l_1\in \{3,\pmb{4},5\}$, $l_2\in\{\pmb{4},5\}$, $l_3\in\{3,\pmb{4},5,6\}$. The numbers on the arcs represent the transition times~$t_e$, $e\in A$. The scenario set is $\mathcal{U}(\Gamma^d)$. The bold path corresponds to lead-time scenario $\pmb{\ell}=(4,5,4)$. }
	\label{fig:network}
\end{figure}

We are given a production plan~$\pmb{x}$.
Using Proposition~\ref{prop_delta} we can construct a layered digraph  $G=(V,A)$ that represents all possible lead-time scenarios. The set  of nodes $V$ is partitioned into  disjoint layers $V_{T+1},\dots, V_{T^+}$ that correspond to future periods $t\in [T^+]$. The nodes of the layer $V_t$ correspond to all possible cumulative deliveries in $\mathcal{O}_t$, described in Proposition~\ref{prop_delta}. Namely, node $v^t_U$ in layer $V_t$ corresponds to period $t\in [T^+]$ and $U$ is the set of summation indices in~(\ref{edelta}), which yields the cumulative delivery $x(U)=\sum_{i\in U} x_i$ at period~$t$. An arc exists from  $v^{t-1}_W$ to $v^t_U$ if $W\subseteq U$. We
 add the source $\mathfrak{s}=v^{T}_{\emptyset}$ linked to all nodes of the layer $V_{T+1}$. By the assumptions imposed in Section~\ref{sec:uncelead}, the layer $V_{T^+}$ contains only one node $v_{[T]}^{T^+}$ (production from all periods in $T$ must arrive by $T^+$).
 We set the sink $\mathfrak{t}=v^{T^+}_{[T]}$.  The cost of arc $(v^{t-1}_W, v^{t}_U)$ is equal to $ \max\{c^I_{t}(x(U)-D_{t}),c^B_{t}(D_{t}-x(U))\}$. The transition time~$t_e$ of the arc $e=(v^{t-1}_W, v^{t}_U)$ is set as follows. Let $Q=\{i\in U\setminus W:  i+\widehat{l}\neq t\}$. For the uncertainty set $\mathcal{U}(\Gamma^d)$, we set the transition time to $|Q|$, for the uncertainty set $\mathcal{U}(\Gamma^c)$ we fix the transition time to $\sum_{i\in Q} |t-(i+\widehat{l})|$. Finally, 
 we set the bound on the total transition time $K=\Gamma$, $\Gamma\in \{\Gamma^d,\Gamma^c\}$.

A sample construction is shown in Figure~\ref{fig:network}. For example, in period $6\in [T^+]$, we have $I_6=1$ (production in period 1 must be delivered by 6) and $J_6=3$ (productions in periods $2,3$ can be delivered by 6). In what follows, layer $V_6$ contains three nodes $v^6_{\{1\}}$, $v^6_{\{1,2\}}$, $v^6_{\{1,2,3\}}$ that represent three possible cumulative deliveries $O_6=x_1$, $O_6=x_1+x_2$, $O_6=x_1+x_2+x_3$ that correspond to three possible lead-time scenarios. The transition time of the arc $(v_{\emptyset}^5,v_{\{1,2,3\}}^6)$ is equal to~2, because $|Q|=|\{1,4\}|=2$.

Each feasible $\mathfrak{s}$-$\mathfrak{t}$ path in $G$ models a possible cumulative delivery plan $\pmb{O}$ implied by
a given production plan~$\pmb{x}$. The cost of this path is the total inventory and backorder cost of $\pmb{x}$ and its
total transition time represents the left-hand side of the budget constraint in $\mathcal{L}(\Gamma)$, $\Gamma\in \{\Gamma^d,\Gamma^c\}$.
(see (\ref{dus}) and (\ref{cus})). For example, the bold path shown in Figure~\ref{fig:network} represents the cumulative delivery plan $\pmb{O}$ such that $O_4=0$, $O_5=x_1$, $O_6=x_1$, $O_7=x_1+x_2+x_3$, $O_8=x_1+x_2+x_3$, $O_9=x_1+x_2+x_3$. This plan corresponds to  lead-time scenario $l_1=4$, $l_2=5$, and $l_3=4$ for $\Gamma^d=1$ (at most one lead time is allowed to be different from $\widehat{l}=4$).

Since the production and setup costs of $\pmb{x}$ are fixed and do not depend on lead times, a shortest constrained 
$\mathfrak{s}$-$\mathfrak{t}$ path models a best lead-time scenario, while a longest  constrained $\mathfrak{s}$-$\mathfrak{t}$ path models a worst lead-time scenario for given~$\pmb{x}$. The graph $G$ has $O(T\cdot (T^+-T))$ nodes. Observe that $\Gamma^d$ is a nonnegative integer not greater than $T$ and $\Gamma^c$ is a nonnegative integer not greater than $T\cdot (T^+-T)$. 
Hence, a best  (a worst) lead-time scenarios can be computed in polynomial time by applying the $O(|V| K)$ algorithm for 
the \textsc{CSP} (\textsc{CLP}) problem, where $K=\Gamma^c$ or $K=\Gamma^d$.
 We have thus proved the following theorem.
 \begin{thm}
 Computing the best and worst lead-time scenarios for a given production plan
 under 
the budgeted uncertainty sets $\mathcal{L}(\Gamma^d)$ and $\mathcal{L}(\Gamma^c)$ can be done in
polynomial time.
 \end{thm}

\subsection{Computing the pessimistic production plan}
\label{secpessymplan}

We begin by showing that the problem of computing a worst-case lead-time scenario for a given production plan~$\pmb{x}\in \Xset$ boils down to the \emph{longest path problem} in an acyclic graph. Specifically, we transform the \textsc{CLP} problem in the graph $G=(V,A)$, constructed in Section~\ref{secwblead}, into a longest path problem in another acyclic graph $G'=(V',A')$. The construction of~$G'$ is the same for both $\Gamma^c$ and $\Gamma^d$.
Namely,
we split each node $v^t_{U}$ of $G$, $t\in[T^+]$, into $\Gamma+1$ nodes $v^{t,i}_{U}$, $i=0,\dots,\Gamma$, 
where $\Gamma\in \{\Gamma^d,\Gamma^c\}$. An arc $e'\in A'$, $e'=(v^{t,i}_W,v^{t+1,j}_U)$,
exists from $v^{t,i}_W$ to $v^{t+1,j}_U$ if $W\subseteq U$ and $j=i+t_e\leq \Gamma$, where $t_e$ is the transition time of the arc 
$e=(v^t_W, v^{t+1}_U)\in A$ in $G$. The cost the arc is the same as in~$G$, i.e.
it is equal to $ \max\{c^I_{t+1}(x(U)-D_{t+1}),c^B_{t+1}(D_{t+1}-x(U))\}$.
We add the sink $\mathfrak{t}$ to $G'$ and link all the nodes from the last layer to it using dummy (dashed) arcs. The cost of all these dummy arcs are~0.
We remove from $G'$ all nodes that are not reachable from the source~$\mathfrak{s}=v^{T,0}_{\emptyset}$ and from  which the sink~$\mathfrak{t}$ is not reachable. A sample construction is shown in Figure~\ref{fig:network-trans}. Note that node $v^{7,2}_{\{1,2\}}$ is 
not present in $G'$, since there is no path from this node to the sink~$\mathfrak{t}$.
 \begin{figure}[ht]
	\centering
	\includegraphics[scale=0.75]{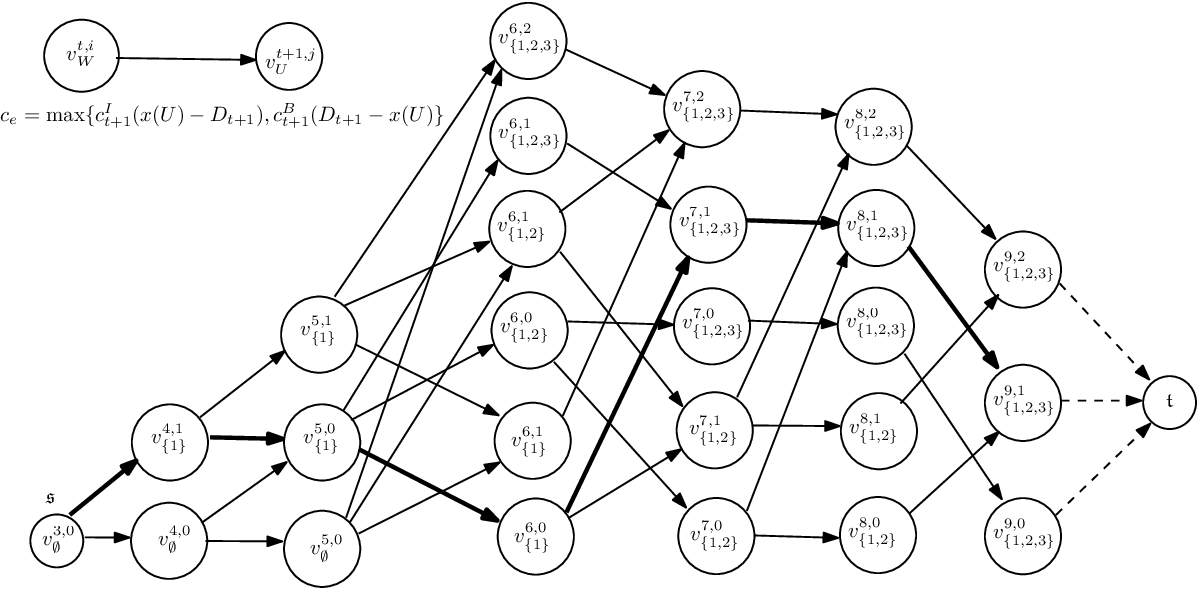}
	\caption{The transformed graph $G'$ for the example shown in Figure~\ref{fig:network} with $\Gamma=2$. The bold path corresponds to lead-time scenario $l_1=4$, $l_2=5$, $l_3=4$. }
	\label{fig:network-trans}
\end{figure}
Now, it is evident that each $\mathfrak{s}$-$\mathfrak{t}$ path in $G'$ corresponds to some feasible lead-time scenario 
under the uncertainty set~$\mathcal{L}(\Gamma)$, $\Gamma\in \{\Gamma^d,\Gamma^c\}$, and vise-versa, each feasible lead-time scenario corresponds to some 
$\mathfrak{s}$-$\mathfrak{t}$ path in $G'$.
The cost of a longest path is the largest total inventory and the backorder cost of a given~$\pmb{x}$ under a worst lead-time scenario. 

We now formulate a Mixed-Integer Programming (MIP) model to solve the 
\textsc{Min-Max} problem (see \eqref{probpes}), i.e.
to compute a pessimistic production plan $\pmb{x}^p$. The MIP is constructed based on the graph $G' = (V', A')$. 

Let $z_{U}^{t,i} \geq 0$ and $z_{\mathfrak{t}} \geq 0$ be variables defined for each node of $G'$. Note that this MIP is 
derived from the dual of the longest path formulation in $G'$, where the components of $\pmb{x}$ embedded in the arc costs of $G'$ are now treated as decision variables. 
Consequently, a pessimistic production plan $\pmb{x}^p$ can be obtained by solving the following MIP formulation:
\begin{align}
 \min \;  & z_{\mathfrak{t}}+\sum_{i\in [T]} (c^P_t x_t+c^S_ty_t) \label{wp1} \\ 
 \text{ s.t. }	 & z_{\mathfrak{t}} \geq z_{v} & v\in V' \label{wp2} \\ 
	 & z_{U}^{t+1,j}-z_{W}^{t,i}\geq c^I_{t+1}(x(U)-D_{t+1}) & (v_{W}^{t,i}, v_{U}^{t+1,j})\in A'\\
	  & z_{U}^{t+1,j}-z_{W}^{t,i}\geq c^B_{t+1}(D_{t+1}-x(U)) & (v_{W}^{t,i}, v_{U}^{t+1,j})\in A' \label{wp3}\\
	  & x_t\leq C_ty_t & t\in [T]\\
	  & y_t\in \{0,1\} & t\in [T]\\
	  & z_v\geq 0 & v\in V'\\
	  & \pmb{x}\in \Xset\subseteq \Rset^{T}_{+}, &  \label{wp8}
 \end{align}
 where $x(U)=\sum_{i\in U} x_i$.

The \textsc{Min-Max} problem is NP-hard, as it reduces to the deterministic lot-sizing problem \eqref{spp1}--\eqref{spp7} by setting $\Gamma^d = \Gamma^c = 0$, which is known to be NP-hard \cite{florian1971deterministic, FLK80}. However, observe that our formulation \eqref{wp1}--\eqref{wp8} contains only $T$ binary variables, which are associated with the setup costs for periods in $[T]$. This observation leads to the following corollary.

\begin{cor}
If all setup costs $c_t^S$ for $t \in [T]$ are equal to 0, then the pessimistic production plan $\pmb{x}^p$ can be computed by solving a linear programming model. Equivalently, the \textsc{Min-Max} problem without setup costs is polynomially solvable under the budgeted uncertainty sets $\mathcal{L}(\Gamma^d)$ and $\mathcal{L}(\Gamma^c)$.
\end{cor}

\subsection{Computing the optimistic production plan}
\label{secoptymplan}
In this section, we first present a mixed-integer programming formulation 
for the \textsc{Min-Min} problem (see \eqref{probopt}), based on a 
min-cost flow model. We then develop a pseudopolynomial-time algorithm 
for this problem.

\subsubsection{Mixed integer programming formulation}
\label{secoptymplan1}

We recall that computing a worst-case lead-time scenario for a given production plan~$\pmb{x}$ 
under the budgeted uncertainty sets $\mathcal{L}(\Gamma^d)$ and $\mathcal{L}(\Gamma^c)$ 
reduces to the \textsc{CSP} in an acyclic graph (see Section~\ref{secwblead}). 
In the same way as in Section~\ref{secpessymplan}, we can construct the corresponding 
acyclic graph~$G'$, in which each $\mathfrak{s}$-$\mathfrak{t}$ path corresponds 
to a feasible lead-time scenario from the uncertainty set~$\mathcal{L}(\Gamma)$, 
$\Gamma \in \{\Gamma^d, \Gamma^c\}$. Conversely, each feasible lead-time scenario 
corresponds to an $\mathfrak{s}$-$\mathfrak{t}$ path in~$G'$, and the cost 
of a shortest path represents the smallest inventory and backorder cost 
for the production plan~$\pmb{x}$.
Unfortunately, using $G'$ to compute an optimistic production plan—where the components of $\pmb{x}$ embedded in the arc costs of $G'$ are treated as decision variables—would lead to a complex MIP formulation, involving a large number of binary variables. In this section, we propose a MIP formulation for computing an optimistic production plan based on a min-cost flow reformulation of the single-item  lot sizing problem (see, e.g.,~\cite{AMO93}).

 We construct a graph $G=(V,A)$ as follows. We introduce a source node $\mathfrak{s}$, nodes $u_t$ for $t \in [T]$, and nodes $v_t$ for $t \in [T^+]$. The source $\mathfrak{s}$ has a supply of $\sum_{t \in [T]} d_t$, while each node $v_t$ has a prescribed demand $-d_t$ for $t \in [T^+]$. For each $t \in [T]$, we add an arc $(\mathfrak{s}, u_t)$ with cost $c^P_t$ and capacity $C_t$; these arcs represent production in periods $t \in [T]$. 
For each $t \in [T]$ and $l \in \{l_t^{\min}, l_t^{\min}+1, \ldots, l_t^{\max}\}$, we add a \emph{lead-time arc} $(u_t, v_{t+l})$ with a cost of 0 and capacity $C_t$. Selecting exactly one lead-time arc emanating from each $u_t$ for $t \in [T]$ corresponds to a lead-time scenario in the uncertainty set $\mathcal{L}$.
Finally, for each $t \in [T^+ - 1]$, we add an arc $(v_t, v_{t+1})$ with cost $c_t^I$ and a sufficiently large capacity $M = \sum_{t \in [T]} C_t$, and an arc $(v_{t+1}, v_t)$ with cost $c_t^B$ and capacity $M$. These arcs represent inventory and backordering between consecutive future periods, respectively. To ensure feasibility, we add an arc $(\mathfrak{s}, v_{T^+})$ with cost $c_{T^+}^B$ and capacity $M$; this arc is utilized, in particular, when $\sum_{t \in [T]} C_t < \sum_{t \in [T^+]} d_t$. Note that in any optimal production plan, we can assume that the inventory level in the final period $T^+$ is zero. Consequently, there is no arc that represents the inventory beyond the period $T^+$. A sample construction is illustrated in Figure~\ref{fig:opt-network}.
\begin{figure}[ht]
	\centering
	\includegraphics[height=7cm]{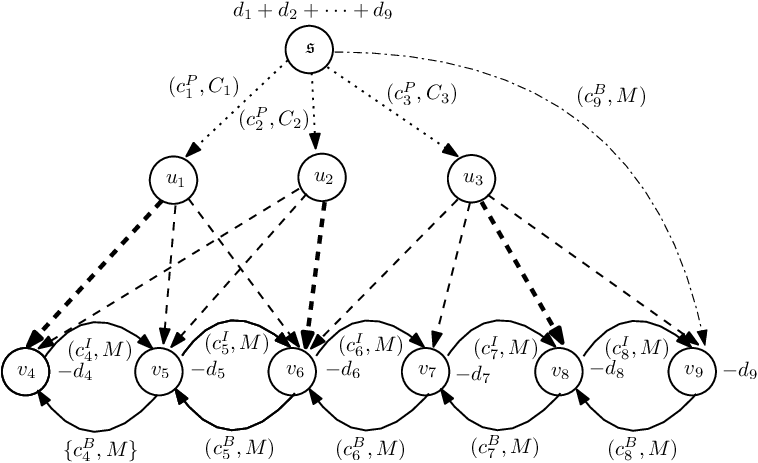}
	\caption{A sample network for $T=3$, $T^+=9$, $\pmb{x}=(x_1,x_2,x_3)$, $\widehat{l}=\pmb{4}$,  
	$\zeta^+_1=\zeta^-_1=1$, 
         $\zeta^+_2=0$, $\zeta^-_2=2$,
	$\zeta^+_3=2$, $\zeta^-_2=1$,
	$l_1\in \{3,\pmb{4},5\}$, $l_2\in\{2,3,\pmb{4}\}$, $l_3\in \{3,\pmb{4},5,6\}$. The lead-time scenario shown in bold is $l_1=3$, $l_2=4$, $l_3=5$.}
	\label{fig:opt-network}
\end{figure}

Let $\mathcal{P}_t = \{t + l_t^{\min}, \dots, t + l_t^{\max}\}$ for $t \in [T]$ be the set describing the possible arrival times for the production initiated in period $t$. Furthermore, let $f(u,v) \geq 0$ denote the flow on arc $(u,v) \in A$, and let $c(u,v)$ be the corresponding cost of the arc. 

We are now ready to present a mixed-integer program  for the \textsc{Min-Min} problem (see \eqref{probopt}) based on a min-cost flow model. Specifically, an optimistic production plan $\pmb{x}^{o}$ under the budgeted uncertainty set $\mathcal{L}(\Gamma^d)$ can be obtained by solving the following MIP model:
\begin{align}
	\min & \sum_{(u,v)\in A} c(u,v)\cdot f(u,v)+\sum_{t\in [T]} c_t^S y_t \label{mobjop}\\
	     \text{s.t. }   & \sum_{j\in \mathcal{P}_t} \rho_{t,j}=1 & t\in [T] \label{mop1}\\
	        & \rho_{t,j}+\rho_{t+1,k}\leq 1 & t\in [T],  j\in \mathcal{P}_{t}, k\in\mathcal{P}_{t+1}: k<j \label{mop2} \\
	        & \sum_{t\in [T], \, j\in \mathcal{P}_t : \, j\neq t+\widehat{l}} \rho_{t,j}\leq \Gamma^d \label{mop2_1} \\
	        & f(\mathfrak{s},u_t)=x_t &  t\in [T] \label{mop3}\\
	        & f(u_t, v_j)\leq \rho_{t,j}C_t &   t\in [T], j\in \mathcal{P}_t \label{mop4} \\
	         & \text{Flow mass balance and capacity constraints} \label{mop5} \\
	         & x_t\leq C_t y_t & t\in [T]\\
	        & y_t\in \{0,1\} & t\in [T] \label{mop6}\\
	        & \rho_{t,j}\in \{0,1\} & t\in [T], j\in \mathcal{P}_t\label{mop7}\\
	        & x_t\geq 0 & t\in [T],\label{mop8}
\end{align}
where
binary variables $\rho_{t,j}$ model feasible arrival times. 
Namely, $\rho_{t,j}=1$ if production $x_t$ in period $t\in [T]$ arrives in period $j\in \mathcal{P}_t$. 
The constraints~(\ref{mop1}) ensure that the production initiated in $t \in [T]$ arrives in exactly one of the possible future periods within $\mathcal{P}_t \subseteq [T^+]$. The constraints~(\ref{mop2}) ensure the chronological consistency of arrivals; specifically, that production from period $t$ cannot arrive after the production from period $t+1$, i.e.
the order crossovers are not allowed. In the example shown in Figure~\ref{fig:network}, where $\mathcal{P}_1 = \{4,5,6\}$ and $\mathcal{P}_2 = \{4,5,6\}$, we add the constraints $\rho_{1,5} + \rho_{2,4} \leq 1$, $\rho_{1,6} + \rho_{2,4} \leq 1$, and $\rho_{1,6} + \rho_{2,5} \leq 1$ to prevent overlapping or inverted lead times. The constraint~(\ref{mop2_1}) models the budget constraint of the uncertainty set $\mathcal{L}(\Gamma^d)$ (see~(\ref{dus})).
In the example, we get 
$$\rho_{1,4}+\rho_{1,6}+\rho_{2,4}+\rho_{2,5}+\rho_{3,6}+\rho_{3,8}+\rho_{3,9}\leq \Gamma^d.$$
To model 
the budget constraint in $\mathcal{L}(\Gamma^c)$
(see (\ref{cus})), we replace~(\ref{mop2_1}) with
$$
\sum_{t\in [T], \, j\in \mathcal{P}_t} |j-t-\widehat{l}|\rho_{t,j}\leq \Gamma^c,
$$
which in the example yields
$$
\rho_{1,4}+\rho_{1,6}+2\rho_{2,4}+\rho_{2,5}+\rho_{3,6}+\rho_{3,8}+2\rho_{3,9}\leq \Gamma^c.
$$
The constraints~(\ref{mop3}) fix the flow on each arc $(\mathfrak{s}, u_t)$ to be equal to the production quantity $x_t$ for $t \in [T]$. The constraints~(\ref{mop4}) ensure that the production $x_t$ is directed to the correct future period $j \in [T^+]$, determined by the lead-time scenarios. Finally, the constraints~(\ref{mop5}) represent the standard mass balance and arc capacity flow constraints for the network~$G$ (see, e.g.,~\cite{AMO93}).
For each feasible choice of $\rho_{t,j}$, the corresponding flow $f(u,v)$ models a delivery plan with the total production, inventory, and backorder cost equal to the cost of the flow. Adding the setup costs, we get the total cost of the production plan $\pmb{x}$. Hence, an optimal solution to~(\ref{mop1})-(\ref{mop8}) represents the cheapest (optimistic) production plan $\pmb{x}^o$, under a best lead-time scenario.

\subsubsection{Pseudopolynomial  algorithm}
\label{spsedo}

We now show that the optimistic production plan can be computed in pseudopolynomial time, that is the \textsc{Min-Min} problem is pseudopolynomially solvable,
under the uncertainty set~$\mathcal{L}(\Gamma)$, 
$\Gamma \in \{\Gamma^d, \Gamma^c\}$.
 Note first that if we fix the optimal values of  binary variables $\pmb{\rho}$ and $\pmb{y}$ in (\ref{mobjop})-(\ref{mop8})), then we can find an optimal production plan $\pmb{x}$ by solving a min-cost flow problem (see also Figure~\ref{fig:opt-network} - we remove the arcs that correspond to $\rho_{t,j}=0$, fix $C_t=0$ if $y_t=0$, and solve the resulting min-cost flow problem). It is well known (see, e.g.,~\cite{AMO93}) that the min-cost flow problem has integral optimal solution provided that node supply/demands and arc capacities are integral. This leads to the following property:
\begin{prop}
\label{propint}
If all capacities $C_t$, $t\in [T]$, and demands $d_t$, $t\in [T^+]$, are integers, then there is an integral optimistic production plan such that $x_t\in \{0,1,\dots,C_t\}$ for each $t\in [T]$.
\end{prop}

To compute an optimistic production plan, we modify the construction shown in Section~\ref{secwblead}. We split each node $v^t_W$ of the graph $G=(V,A)$ into nodes $v^t_W(i)$, where $i=0,1,\dots,\sum_{k\in W} C_k$ ($i=0$ if $U=\emptyset$). In the following, $V$ denotes the extended set of nodes. We then create the set of arcs in the following way. For each node $v^{t-1}_W(i)\in V$ we add an arc that leads to node $v^{t}_U(j)\in V$ if $W\subseteq U$ and $j\in \{i,i+1,\dots,i+\sum_{k\in U\setminus W} C_k\}$.  This arc represents a delivery of $j-i$ units from periods in $U\setminus W$ that arrives in period $t$, which corresponds to the equality $\sum_{k\in U\setminus W} x_k=j-i$. The transition time of the arc $(v^{t-1}_W(i), v^{t}_U(j))$ is the same as $(v^{t-1}_W, v^{t}_U)$ in the original graph, as it depends only on $t$, $W$ and $U$.  Finally, we add a sink~$\mathfrak{t}$ and arcs $(v_W^{T^+},\mathfrak{t})$ for each node $v_W^{T^+}\in V$. A sample construction  is shown in Figure~\ref{fig:pseudopoly-network}. For example, the bold arc $(v^5_{\{1\}}(1), v^6_{\{1,2,3\}}(3))$ represents delivery of~2 units from periods 2 and 3 to period 6. The transition time of this arc is equal to~1, since the common lead time is equal to~4 and $l_2=4$, $l_3=3$.
 \begin{figure}
	\centering
	\includegraphics[scale=0.73]{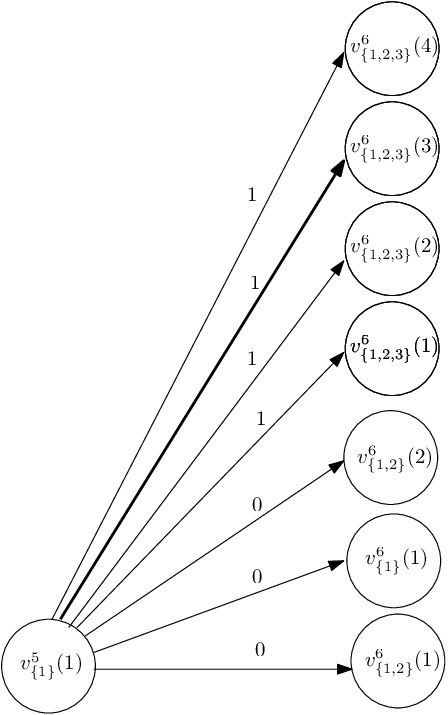}
	\caption{A sample construction for the data from Figure~\ref{fig:network-trans} and $C_1=2$, $C_2=1$, $C_3=2$. The bold arc represents delivery of $2$ units from periods $2$ and $3$ to period $6$. The number on arcs represent the transition times.}
	\label{fig:pseudopoly-network}
\end{figure}
    
It remains to fix the costs of the arcs. Observe that for node $v^t_W(i)$, $i$ is the cumulative delivery in period $t$. Hence, the cost of the arc $(v^{t-1}_W(i), v^{t}_U(j))$ is equal to
$$\max\{c^I_{t}(j-D_{t}),c_{t}^B(D_{t}-j)\}+c_{j-i}^*(W,U),$$
where $c_{j-i}^*(W,U)$ is the optimal cost of delivery $j-i$ units from periods in $U\setminus W$. The value of $c_{j-i}^*(W,U)$ can be computed by solving the following problem:
\begin{equation}
\label{subprob}
\begin{array}{lllll}
	c_{j-i}^*(W,U)= & \min & \displaystyle\sum_{t\in U\setminus W} (c^P_t x_t + c^S_t y_t) \\
	 &\text{ s.t. }& \displaystyle \sum_{t\in U\setminus W} x_t=j-i \\
	&& 0\leq x_t\leq y_t C_t & t\in U\setminus W\\
	&& y_t\in \{0,1\} &  t\in U\setminus W.
\end{array}
\end{equation}
The problem~(\ref{subprob}) is a  \emph{continuous knapsack problem  with setups}. It is weakly NP-hard for arbitrary setup costs. Since it is a special case of the single-item  lot sizing problem with concave costs and without backordering, it can be solved in pseudopolynomial time $O(T\cdot C)$, where $C=\sum_{t\in U\setminus W} C_t$ (see~\cite{FLK80}). If all setup costs are equal to 0, then~(\ref{subprob}) becomes a \emph{continuous knapsack problem} that can be solved in polynomial time $O(T)$ by a simple greedy algorithm.
As in Section~\ref{secwblead} we set a bound on the total transition time $K$ to~$\Gamma$, 
$\Gamma \in \{\Gamma^d, \Gamma^c\}$, and solve the CSP problem in the resulting graph. 
\begin{prop}
	There are integral solution $\pmb{x}\in \Xset$ and $\pmb{\ell}\in \mathcal{L}(\Gamma)$, 
$\Gamma \in \{\Gamma^d, \Gamma^c\}$,  to~(\ref{probopt}) with the total cost at most $c$ if and only if there is a feasible 
$\mathfrak{s}$-$\mathfrak{t}$ path $P$ in $G$ with the total cost at most $c$.
\end{prop}
\begin{proof}
Let $\pmb{x}=(x_1,\dots,x_T)\in \Xset$, $\pmb{\ell}=(l_1,\dots,l_T)\in \mathcal{L}(\Gamma)$ be a feasible integral solution to~(\ref{probopt}) with the total cost $c$. Choose $\pmb{y}\in\{0,1\}^T$, where $y_t=1$ iff $x_t>0$ for $t\in [T]$.
We  form the corresponding path $P$ in $G$ as follows. Define a sequence of sets $U_T,U_{T+1},\dots,U_{T^+}$ where  $U_t=\{i\in [T]: i+l_i\leq t\}$, $t=T,T+1,\dots,T^+$. The set $U_t$ contains all periods in $[T]$ from which the production arrived in period $t\in [T^+]$. We get $U_T=\emptyset$ and $U_T\subseteq U_{T+1}\subseteq\dots\subseteq U_{T^+}$. Define $O_t=\sum_{i\in U_t} x_i$ for $t=T,T+1,\dots,T^+$, so $O_t$ is the cumulative delivery in period $t=T,T+1,\dots,T^+$ (notice that $O_T=0$). Let  $P$ be the path composed of nodes $v^t_{U_t}(O_t)$ for $t=T,T+1,\dots,T^+$ and the sink~$\mathfrak{t}$. By the construction, the path $P$ exists in $G$ and is feasible, as its transition time does not exceed~$\Gamma$. The cost of this path is
$$c(P)=\sum_{t\in [T^+]}\max\{c^I_{t}(O_t-D_{t}),c_{t}^B(D_{t}-O_t)\}+c_{O_{t}-O_{t-1}}^*(U_{t-1},U_t).$$
By the construction (see the problem~(\ref{subprob}))
$$c_{O_{t}-O_{t-1}}^*(U_{t-1},U_t)\leq \sum_{i\in U_{t+1}\setminus U_t}  c^P_t x_t+c^S_t y_t.$$
Since
$$\sum_{t\in [T^+]} \sum_{i\in U_{t+1}\setminus U_t} c^P_t x_t+c^S_t y_t=\sum_{t\in [T^+]} c^P_tx_t+c^S_ty_t,$$
we get $c(P)\leq c$. 

Suppose that there is a feasible path $P$ in $G$ with the cost $c(P)$. This path goes through nodes $v^T_{\{\emptyset\}}(0), v^{T+1}_{U_{T+1}}(O_1),\dots,v^{T^+}_{U_{T^+}}(O_{T^+}),\mathfrak{t}$, where $U_{T}\subseteq U_{T+1}\subseteq\dots \subseteq U_{T^+}$ and $O_{T}\leq O_{T+1}\leq \dots\leq O_{T^+}$, by the construction of $G$. We construct a feasible production plan as follows. For each $t\in [T^+]$ and each $i\in U_t\setminus U_{t-1}$, we set the lead time $l_i=t-i$ and $x_i$, $y_i$ to the values of the optimal solution to~(\ref{subprob}) corresponding to $c_{O_{t}-O_{t-1}}^*(U_{t-1},U_t)$. By the construction of $G$ and~(\ref{subprob}), $\pmb{\ell}\in \mathcal{L}(\Gamma)$, $0\leq x_i\leq b_i$ and $y_i=1$ iff $x_i>0$. Therefore, the production plan $\pmb{x}$ is feasible and its cost is equal to $c(P)$. So, the cost of an optimal solution to~(\ref{probopt}) is at most $c(P)$.
\end{proof}

It remains to show that the resulting CSP problem can be solved in pseudopolynomial time.  Let $C=\sum_{t\in [T]} C_t$ be the total capacity. Graph $G$ has at most  $T^+\cdot T \cdot C$ nodes and at most $T^+\cdot (T \cdot C)^2$ arcs. The network $G$ with all arc-costs can be built in $O(T^+\cdot (T \cdot C)^3)$ time. Finally, the CSP problem can be solved in $O(|A| \Gamma)=O(T^+\cdot (T \cdot C)^2\Gamma)$ time, where $\Gamma \in \{\Gamma^d, \Gamma^c\}$ is a polynomial of $T$ and $T^+$. This leads us to the following theorem.
\begin{thm}
The \textsc{Min-Min} problem is pseudopolynomially solvable under the budgeted uncertainty sets $\mathcal{L}(\Gamma^d)$ and $\mathcal{L}(\Gamma^c)$.
\end{thm}

\begin{rem}
The constrained longest path in $G$ does not correspond to the pessimistic production plan~$\pmb{x}^p$, since the integrality assumption does not hold in this case (see Section~\ref{sec:unclead}, Remark~\ref{remint}).
\end{rem}

The above remark shows that a similar pseudopolynomial algorithm cannot be proposed for solving the \textsc{Min-Max} problem
under $\mathcal{L}(\Gamma^d)$ and $\mathcal{L}(\Gamma^c)$.

\subsection{Computing the \textsc{R*} production plan}

We are now ready to present a method for solving the \textsc{R*} problem, 
which is the primary focus of this paper. Corollary~\ref{corcompl} 
shows that this problem, without backordering and with equal capacities, 
is NP-hard under the budgeted uncertainty sets $\mathcal{L}(\Gamma^d)$ 
and $\mathcal{L}(\Gamma^c)$. Furthermore, a pseudopolynomial algorithm 
similar to the one presented in Section~\ref{spsedo} cannot be 
constructed for the \textsc{R*} problem, as it lacks the integrality 
property (see Remark~\ref{remint}). Accordingly, we propose a 
MIP formulation for the \textsc{R*} problem.
The MIP  for computing  the \textsc{R*} production plan~$\pmb{x}^r$ combines the models constructed in Sections~\ref{secpessymplan} and~\ref{secoptymplan1}. 
 It uses two networks: $G'=(V',A')$,  defined in Section~\ref{secpessymplan}, and $G=(V,A)$, described in Section~\ref{secoptymplan1}. The resulting MIP formulation is as follows:
\begin{align}
	\min & \sum_{(u,v)\in A} c(u,v)\cdot f(u,v)+\sum_{t\in [T]} c_t^S y_i \label{mobjR}\\
	\text{s.t.}\; & \displaystyle z_{\mathfrak{t}} + \sum_{t\in T} (c_t^S y_t+c_t^P x_t)\leq B \label{mR1}\\
	 & \text{Constraints}\;  (\ref{wp2})-(\ref{wp3}) \text{ and}~(\ref{wp8}) \label{mR1a}\\  
	& \text{Constraints}\;  (\ref{mop1})-(\ref{mop5}) \\
	  & x_t\leq  C_t y_t & t\in [T]\\     
	        & y_t\in \{0,1\} & t\in [T]\\
	        & \rho_{t,j}\in \{0,1\} & t\in [T], j\in \mathcal{P}_t\\
	        & x_t\geq 0 & t\in [T] \label{mR2} \\
	        & z_{v}\geq 0 & v\in V'. \label{mR3}
\end{align}
If the problem is infeasible (i.e. the constraint~(\ref{mR1}) cannot be satisfied), then we solve the
\textsc{Min-Max} problem~(\ref{probpes}) and choose the pessimistic production plan $\pmb{x}^p$.

\section{Computational tests}
\label{sec:tests}


This section presents the results of computational tests demonstrating that the use of the~\textsc{R*} criterion can significantly enlarge the set of candidate production plans compared to those resulting from solving the $\textsc{Min-Max}$ and $\textsc{Min-Min}$ 
problems. The experiments were implemented using Julia language
 and JuMP framework~(\cite{JuMP}) 
 \footnote{The complete Julia source code, along with detailed instructions, is available on GitHub at 
\href{https://github.com/KatJon/lotsizing-model}{\texttt{github.com/KatJon/lotsizing-model}}.}. 
The solver used for the tests was Gurobi 12.0.1 (\cite{gurobi}). The time limit for each MIP formulation was set to 300 s. However, all the tested instances have been solved to optimality within this time.

\subsection{Data generation}

The procedure for instance generation was inspired by the instances considered by~\cite{TON22} and~\cite{TA17}. The instances were generated for the number of planning periods $T=10$, $T=15$ and the number of future periods $T^+=21$, $T^+=31$, respectively.  All input data are integers, generated uniformly at random in some ranges $[a:b]$. We chose the following ranges for particular parameters:
\begin{itemize}
	\item $c^P_t\in [10:50]$,  $c^S_t\in [1000:5000]$, $t\in [T]$,
	\item $C_t\in [CAP: 1.5 \cdot CAP]$, $t\in [T]$, where $CAP=\frac{(T^+ - T)\cdot750}{T}$,
	\item $d_t \in [75:750]$, $c^I_t\in [5:10]$, $c^B_t\in [50:100]$, $t\in [T^+]$; additionally, the backordering cost in the last period $T^+$ was overridden to a big constant $10^6$,
	\item $\hat{l}=T$,  $\zeta^+_t\in [2:\frac{T}{2}]$, $\zeta^-_t\in [1, \frac{T}{4}]$, $t\in [T]$; the bounds were adjusted to satisfy the assumptions from Section~\ref{sec:uncelead}.
\end{itemize}

In our experiments, we compared the pessimistic (robust), optimistic, and \textsc{R*} production plans resulting from solving the \textsc{Min-Max}, \textsc{Min-Min}, and \textsc{R*} problems, respectively (see (\ref{probopt}), (\ref{probpes}), and (\ref{probr})). 
To solve these problems, we employed the models (\ref{wp1})-(\ref{wp8}) for \textsc{Min-Max}, (\ref{mobjop})-(\ref{mop8}) for
 \textsc{Min-Min}, and (\ref{mobjR})-(\ref{mR3}) for~\textsc{R*}.

In the case of \textsc{R*} we chose $B=\alpha \cdot \overline{C}$, where $\overline{C}$ is the optimal objective value of the \textsc{Min-Max} problem. We fixed $\alpha\in [1,1.5]$ and denoted the resulting instance as \textsc{R*}$(\alpha)$.  Furthermore, we only tested the budget $\Gamma^d$, whose value was chosen from the range $[1:10]$. We used $\Gamma^d_1=\Gamma^d_2=\Gamma^d$ in~(\ref{probr}). We also considered the nominal production plans in which the lead times $l_t=\hat{l}$ for each $t\in [T]$.

\subsubsection{Analysis of the costs of the production plans}

Let $\pmb{x}=(x_t)_{t\in [T]}$ be a production plan resulting from solving the problem $\textsc{Min-Max}$, $\textsc{Min-Min}$, or $\textsc{R}^*$. To evaluate $\pmb{x}$ we performed a Monte Carlo simulation. We assumed that each lead-time $l_t$, $t\in [T]$, is uniformly distributed in its range $[\hat{l}_t-\underline{\zeta}_t: \hat{l}_t+\overline{\zeta}_t]$. We generated a large sample of lead-time scenarios in $\mathcal{L}$ using the rejection sampling method (infeasible lead-time scenarios that did not meet the assumption of no order crossover were rejected). We created an empirical distribution for the costs of $\pmb{x}$. In the absence of other information, this distribution is a complete characterization of the quality of $\pmb{x}$. Having a family of production plans with such empirical distributions, decision-makers can ultimately choose one of them taking into account their preferences. 

The empirical distributions for sample instances are shown in Figure~\ref{fig:RG1_combined_10_21}  and Figure~\ref{fig:RG1_combined_15_31}.  The corresponding statistics for these distributions are shown in Table~\ref{tab:combined-stats-10-21} and Table~\ref{tab:combined-stats-15-31}.
 Let us first compare the \textsc{Min-Max} (pessimistic) and \textsc{Min-Min} (optimistic) production plans
 $\pmb{x}^p$ and $\pmb{x}^o$, respectively (the first and last row of histograms). As one can expect, the \textsc{Min-Max} production plans are less risky (they have smaller maximum costs) and their distributions have smaller standard deviations. They also have a smaller mean value and median. On the other hand, the \textsc{Min-Min} production plans allow to reach much smaller costs in favorable scenarios, so they cannot be completely excluded from considerations.  The behavior of the \textsc{R*}-production plans $\pmb{x}^r$ is intermediate between the two extreme cases. Depending on the choice of $\Gamma^d$ and $B$ we get production plans that have smaller maximum cost than the \textsc{Min-Min} plans, but also smaller minimum cost than the \textsc{Min-Max} plans. 
 Therefore, by solving
 the \textsc{R*} problem while varying $\Gamma^d$ and $B$, we expand the set of candidate production plans. It is also interesting to compare the \textsc{Min-Max} and \textsc{R*}-production plans to the nominal plan. One can see in Tables~\ref{tab:combined-stats-10-21} and~\ref{tab:combined-stats-15-31} that the nominal plan is more risky, as its maximum cost and also 0.9-quantile and standard deviation, can be much larger than for the \textsc{Min-Max} and \textsc{R*}-production plans.
\begin{figure}
	\centering
	\includegraphics[width=\textwidth]{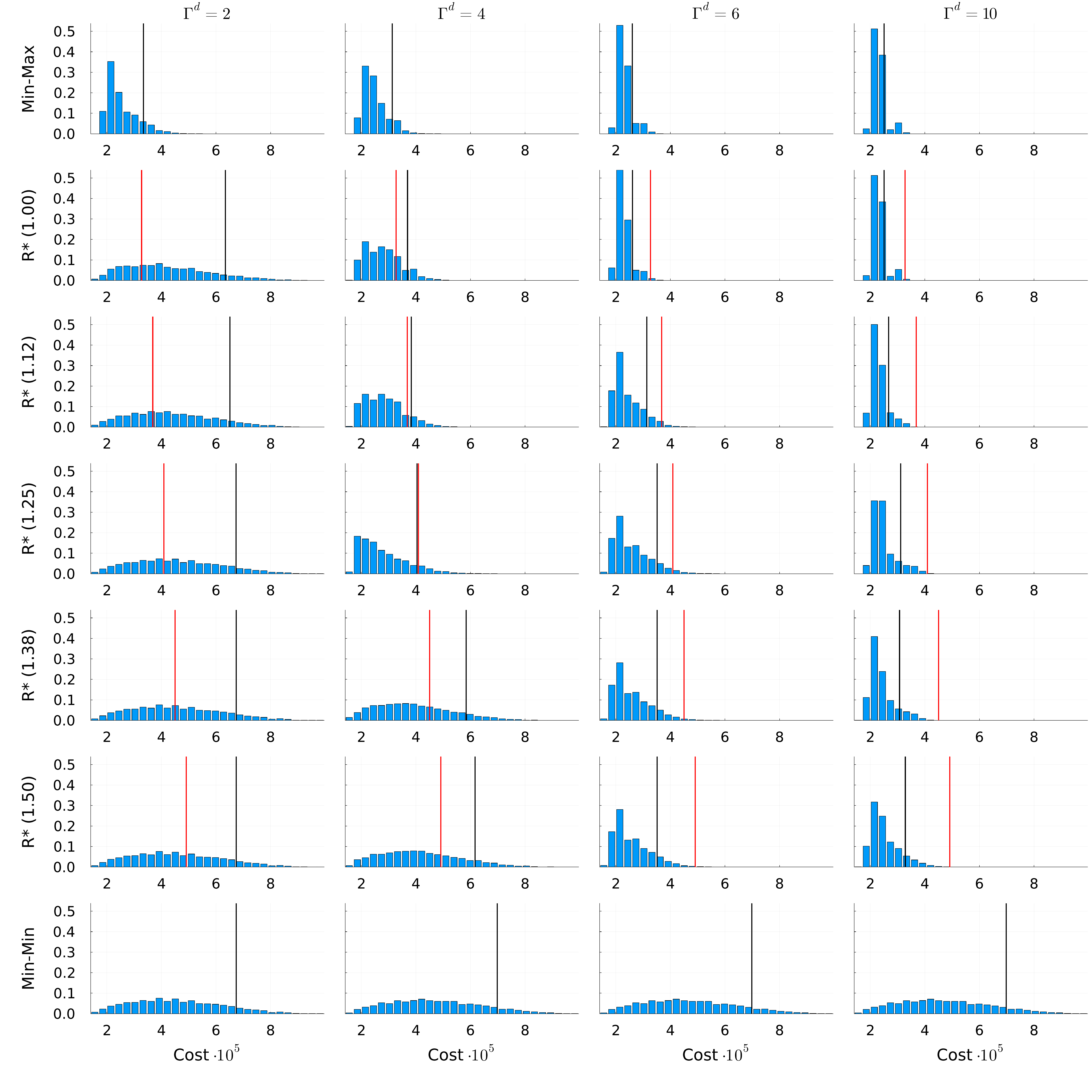}
	\caption{Histograms with the cost distributions of production plans for an instance with $T=10$ and $T^+=21$. The red line represents the budget $B$ and the black line represent the $0.9$-quantile for the cost.}
	\label{fig:RG1_combined_10_21}
\end{figure}
\begin{small}
\begin{table}
	\centering
	\caption{Statistical metrics (the minimum, maximum, mean, median, standard deviation and 0.9-quantile) for the histograms presented in Figure~\ref{fig:RG1_combined_10_21}.}
	\label{tab:combined-stats-10-21}

	{
\small
		\begin{tabular}{ l r *{6}{ |r } }
			\theadl{Problem} & 
			\theadl{$\Gamma^d$} &
			\theadl{Min} &
			\theadl{Max} &
			\theadl{Mean} &
			\theadl{Median} &
			\theadl{Std. dev.} &
			\theadl{0.9-quant} \\
			\hline
			$\textsc{Nominal}$
 & - & 152439.00 & 882468.00 & 399795.49 & 385776.00 & 137522.66 & 594710.00 \\
\hline
\hline
$\textsc{Min-Max}$
 & 2 & 177619.85 & 549088.66 & 252426.35 & 233905.29 & 55544.00 & 334306.30 \\
 & 4 & 176465.95 & 465912.71 & 247526.69 & 238351.85 & 42363.20 & 313333.58 \\
 & 6 & 186942.37 & 368287.02 & 232639.66 & 227272.25 & 25471.09 & 260676.13 \\
 & 10 & 189979.96 & 327500.21 & 231943.66 & 228307.35 & 23567.75 & 251034.21 \\
\hline
\hline
$R^*(1.00)$
 & 2 & 148651.91 & 949719.24 & 419479.18 & 397230.39 & 151554.13 & 634318.50 \\
 & 4 & 162897.40 & 556242.42 & 277741.96 & 270827.49 & 64933.51 & 369439.78 \\
 & 6 & 182986.80 & 363815.37 & 230266.83 & 224977.42 & 26252.21 & 261357.37 \\
 & 10 & 189979.96 & 327500.21 & 231943.66 & 228307.35 & 23567.75 & 251034.21 \\
\hline
\hline
$R^*(1.12)$
 & 2 & 146003.92 & 945319.57 & 436515.54 & 422744.74 & 153367.63 & 650930.05 \\
 & 4 & 159762.00 & 581253.00 & 282929.07 & 273810.00 & 70900.06 & 383371.00 \\
 & 6 & 165100.45 & 519890.30 & 242092.74 & 224412.74 & 50344.38 & 313967.64 \\
 & 10 & 178572.91 & 368437.74 & 231850.66 & 226717.48 & 27908.39 & 267470.09 \\
\hline
\hline
$R^*(1.25)$
 & 2 & 143434.62 & 972211.58 & 455020.37 & 441554.35 & 159867.01 & 673421.26 \\
 & 4 & 157472.00 & 675918.00 & 277769.85 & 256360.00 & 85436.14 & 404332.00 \\
 & 6 & 156554.00 & 578045.00 & 257136.97 & 238254.00 & 65006.07 & 351406.00 \\
 & 10 & 171349.31 & 409375.26 & 247827.40 & 235520.37 & 42267.78 & 311496.66 \\
\hline
\hline
$R^*(1.38)$
 & 2 & 143402.00 & 972553.00 & 455255.30 & 441687.00 & 159957.44 & 673941.00 \\
 & 4 & 152610.13 & 880927.79 & 392168.80 & 376772.15 & 136216.20 & 584303.73 \\
 & 6 & 156554.00 & 578045.00 & 257136.97 & 238254.00 & 65006.07 & 351406.00 \\
 & 10 & 166036.00 & 444287.00 & 241258.29 & 228676.00 & 44603.20 & 307518.00 \\
\hline
\hline
$R^*(1.50)$
 & 2 & 143402.00 & 972553.00 & 455255.30 & 441687.00 & 159957.44 & 673941.00 \\
 & 4 & 150335.00 & 906492.00 & 415151.61 & 401608.00 & 143101.08 & 616711.00 \\
 & 6 & 156554.00 & 578045.00 & 257136.97 & 238254.00 & 65006.07 & 351406.00 \\
 & 10 & 164283.36 & 491250.32 & 251986.99 & 238202.64 & 51414.07 & 328548.33 \\
\hline
\hline
$\textsc{Min-Min}$
 & 2 & 143402.00 & 972553.00 & 455255.30 & 441687.00 & 159957.44 & 673941.00 \\
 & 4 & 141298.00 & 996577.00 & 470611.42 & 456543.00 & 164686.83 & 698250.20 \\
 & 6 & 141298.00 & 996577.00 & 470611.42 & 456543.00 & 164686.83 & 698250.20 \\
 & 10 & 141298.00 & 996577.00 & 470611.42 & 456543.00 & 164686.83 & 698250.20
			\unskip
		\end{tabular}
	}
\end{table}
\end{small}
\begin{figure}
	\centering
	\includegraphics[width=\textwidth]{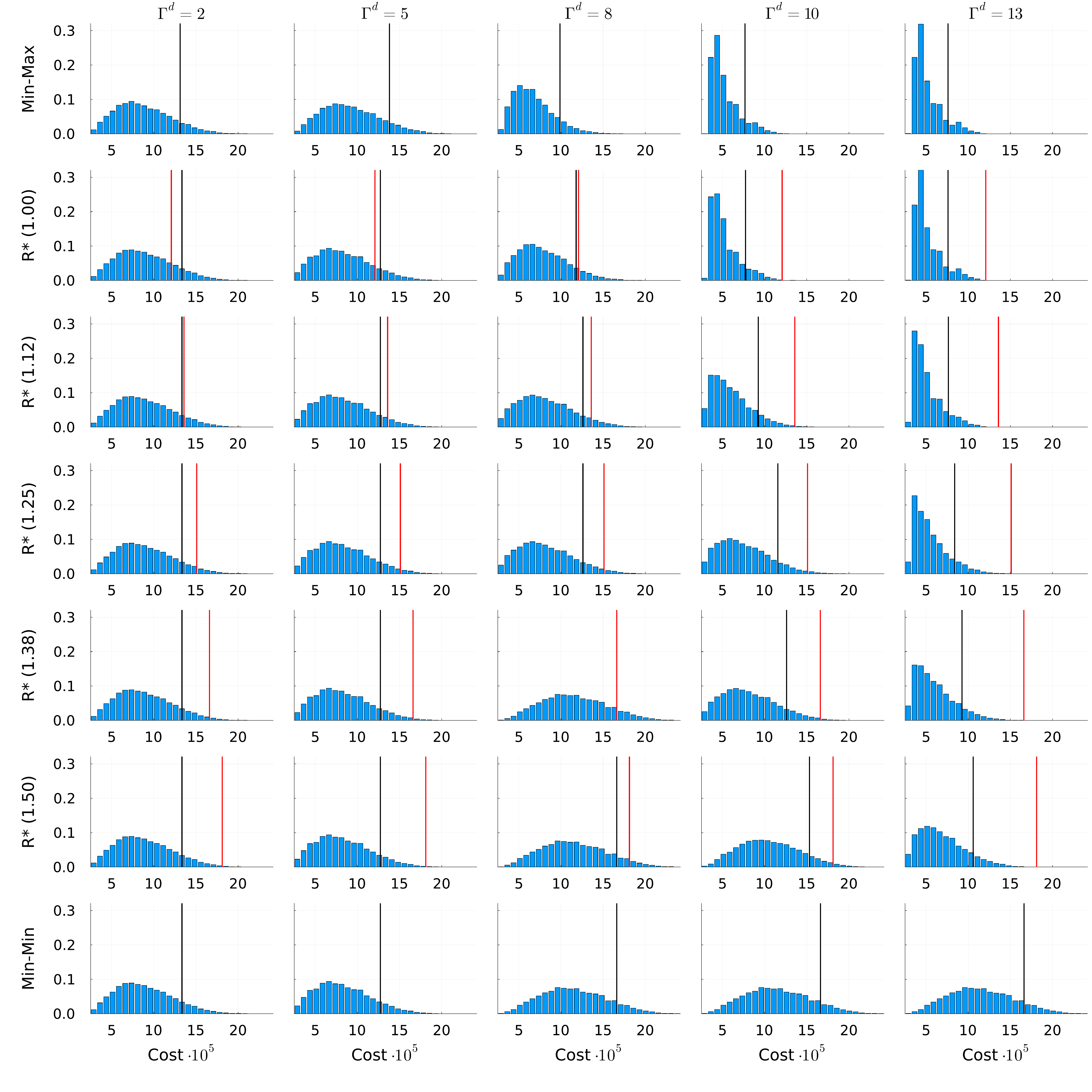}
	\caption{Histograms with the cost distributions of production plans for an instance with $T=15$ and $T^+=31$. The red line represents the budget $B$ and the black line represent the $0.9$-quantile for the cost.}
	\label{fig:RG1_combined_15_31}
\end{figure}
\begin{small}
\begin{table}
	\centering
	\caption{Statistical metrics (the minimum, maximum, mean, median, standard deviation and 0.9-quantile) for the histograms presented in Figure~\ref{fig:RG1_combined_15_31}.}
	\label{tab:combined-stats-15-31}

	{
\small
		\begin{tabular}{ l r *{6}{ |r } }
			\theadl{Problem} & 
			\theadl{$\Gamma^d$} &
			\theadl{Min} &
			\theadl{Max} &
			\theadl{Mean} &
			\theadl{Median} &
			\theadl{Std. dev.} &
			\theadl{0.9-quant} \\
			\hline
			$\textsc{Nominal}$
 & - & 251816.00 & 2095397.00 & 888327.59 & 852974.50 & 325844.53 & 1337602.10 \\
\hline
\hline
$\textsc{Min-Max}$
 & 2 & 253987.01 & 2075191.05 & 871908.29 & 834977.11 & 318867.38 & 1312814.91 \\
 & 5 & 254840.63 & 2117971.84 & 916687.14 & 880564.71 & 333317.94 & 1379948.42 \\
 & 8 & 286201.00 & 1741463.00 & 668204.10 & 632921.50 & 226932.47 & 988407.50 \\
 & 10 & 312704.47 & 1269728.80 & 525794.03 & 471861.85 & 163269.98 & 768383.01 \\
 & 13 & 318785.50 & 1186392.00 & 519458.32 & 460213.59 & 159228.56 & 760241.12 \\
\hline
\hline
$R^*(1.00)$
 & 2 & 251504.00 & 2091913.00 & 886615.91 & 850451.50 & 324902.08 & 1334791.70 \\
 & 5 & 254720.00 & 1965633.00 & 828443.33 & 790921.00 & 318406.41 & 1273498.70 \\
 & 8 & 270320.99 & 1959124.92 & 785164.59 & 747611.36 & 287542.45 & 1179925.43 \\
 & 10 & 307278.09 & 1324871.49 & 523978.54 & 474363.05 & 167829.79 & 773824.80 \\
 & 13 & 318214.94 & 1186465.23 & 519882.18 & 460895.58 & 159184.84 & 760761.34 \\
\hline
\hline
$R^*(1.12)$
 & 2 & 251504.00 & 2091913.00 & 886615.91 & 850451.50 & 324902.08 & 1334791.70 \\
 & 5 & 254720.00 & 1965633.00 & 828443.33 & 790921.00 & 318406.41 & 1273498.70 \\
 & 8 & 250895.00 & 1964040.00 & 816242.89 & 777127.00 & 317464.58 & 1260862.40 \\
 & 10 & 271670.00 & 1659392.00 & 602053.21 & 554402.00 & 226919.59 & 924770.40 \\
 & 13 & 294608.31 & 1314094.26 & 514661.94 & 459964.88 & 168803.39 & 764497.61 \\
\hline
\hline
$R^*(1.25)$
 & 2 & 251504.00 & 2091913.00 & 886615.91 & 850451.50 & 324902.08 & 1334791.70 \\
 & 5 & 254720.00 & 1965633.00 & 828443.33 & 790921.00 & 318406.41 & 1273498.70 \\
 & 8 & 250895.00 & 1964040.00 & 816242.89 & 777127.00 & 317464.58 & 1260862.40 \\
 & 10 & 254101.00 & 1894771.00 & 746543.71 & 703938.00 & 294093.98 & 1157249.30 \\
 & 13 & 287445.94 & 1520735.73 & 550565.25 & 500581.36 & 196580.76 & 839501.68 \\
\hline
\hline
$R^*(1.38)$
 & 2 & 251504.00 & 2091913.00 & 886615.91 & 850451.50 & 324902.08 & 1334791.70 \\
 & 5 & 254720.00 & 1965633.00 & 828443.33 & 790921.00 & 318406.41 & 1273498.70 \\
 & 8 & 275408.00 & 2417560.00 & 1149146.57 & 1125572.00 & 377100.06 & 1662264.40 \\
 & 10 & 250895.00 & 1964040.00 & 816242.89 & 777127.00 & 317464.58 & 1260862.40 \\
 & 13 & 275731.72 & 1636611.23 & 601757.27 & 550109.10 & 227582.92 & 926059.63 \\
\hline
\hline
$R^*(1.50)$
 & 2 & 251504.00 & 2091913.00 & 886615.91 & 850451.50 & 324902.08 & 1334791.70 \\
 & 5 & 254720.00 & 1965633.00 & 828443.33 & 790921.00 & 318406.41 & 1273498.70 \\
 & 8 & 275408.00 & 2417560.00 & 1149146.57 & 1125572.00 & 377100.06 & 1662264.40 \\
 & 10 & 280874.00 & 2201195.00 & 1051427.98 & 1025553.50 & 351735.85 & 1532093.70 \\
 & 13 & 262468.00 & 1749560.00 & 686730.17 & 639779.50 & 263922.92 & 1059096.70 \\
\hline
\hline
$\textsc{Min-Min}$
 & 2 & 251504.00 & 2091913.00 & 886615.91 & 850451.50 & 324902.08 & 1334791.70 \\
 & 5 & 254720.00 & 1965633.00 & 828443.33 & 790921.00 & 318406.41 & 1273498.70 \\
 & 8 & 275408.00 & 2417560.00 & 1149146.57 & 1125572.00 & 377100.06 & 1662264.40 \\
 & 10 & 275408.00 & 2417560.00 & 1149146.57 & 1125572.00 & 377100.06 & 1662264.40 \\
 & 13 & 275408.00 & 2417560.00 & 1149146.57 & 1125572.00 & 377100.06 & 1662264.40
			\unskip
		\end{tabular}
	}
\end{table}
\end{small}

\subsubsection{Two-criteria analysis}

To demonstrate the advantage of the \textsc{R*} approach, we performed a two-criteria analysis. For each production plan $\pmb{x}$ we computed its minimum and maximum cost in the lead-time scenario set $\mathcal{L}$. These quantities and are listed in Table~\ref{tab:combined-stats-10-21} (the columns Min and Max). Figure~\ref{fig:minmax_pareto_10_21} shows these values 
for the \textsc{Min-Min},  \textsc{Min-Max} and \textsc{R*}-production plans, for various $\Gamma^d$. For clarity, we only show the points for instances:
 \textsc{R*}$(1.0)$,  \textsc{R*}$(1.25)$ and  \textsc{R*}$(1.5)$. The shape of the symbols represents instances and the colors represent the values for $\Gamma^d$.
\begin{figure}
	\centering
	\includegraphics[width=\textwidth]{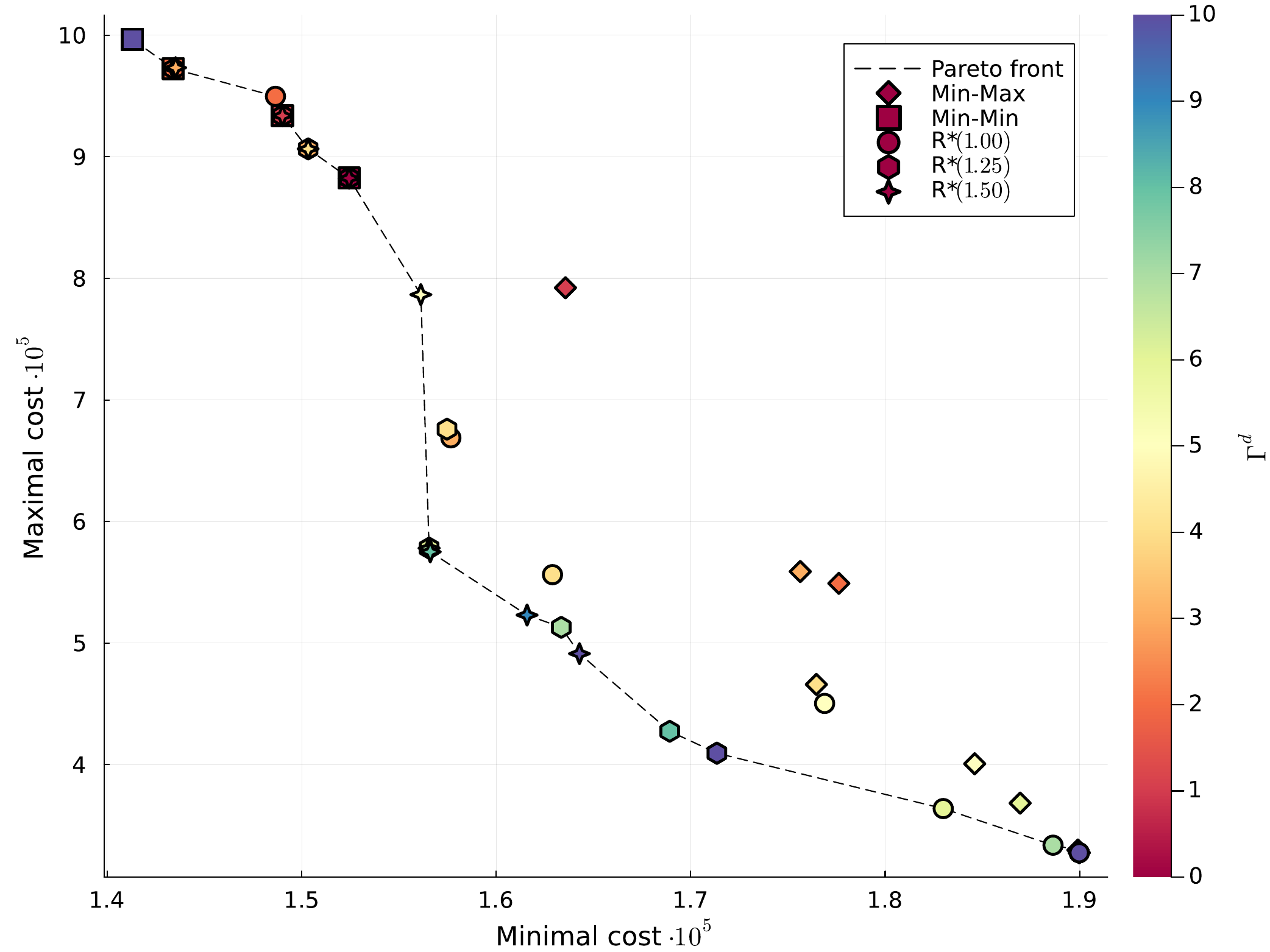}
	\caption{The minimal and maximal  costs of production plans for the instances with $T=10$ and $T^+=21$. The dashed line represents the Pareto front. The colors represent the budget~$\Gamma^d$.}
	\label{fig:minmax_pareto_10_21}
\end{figure}
\begin{figure}
	\centering
	\includegraphics[width=\textwidth]{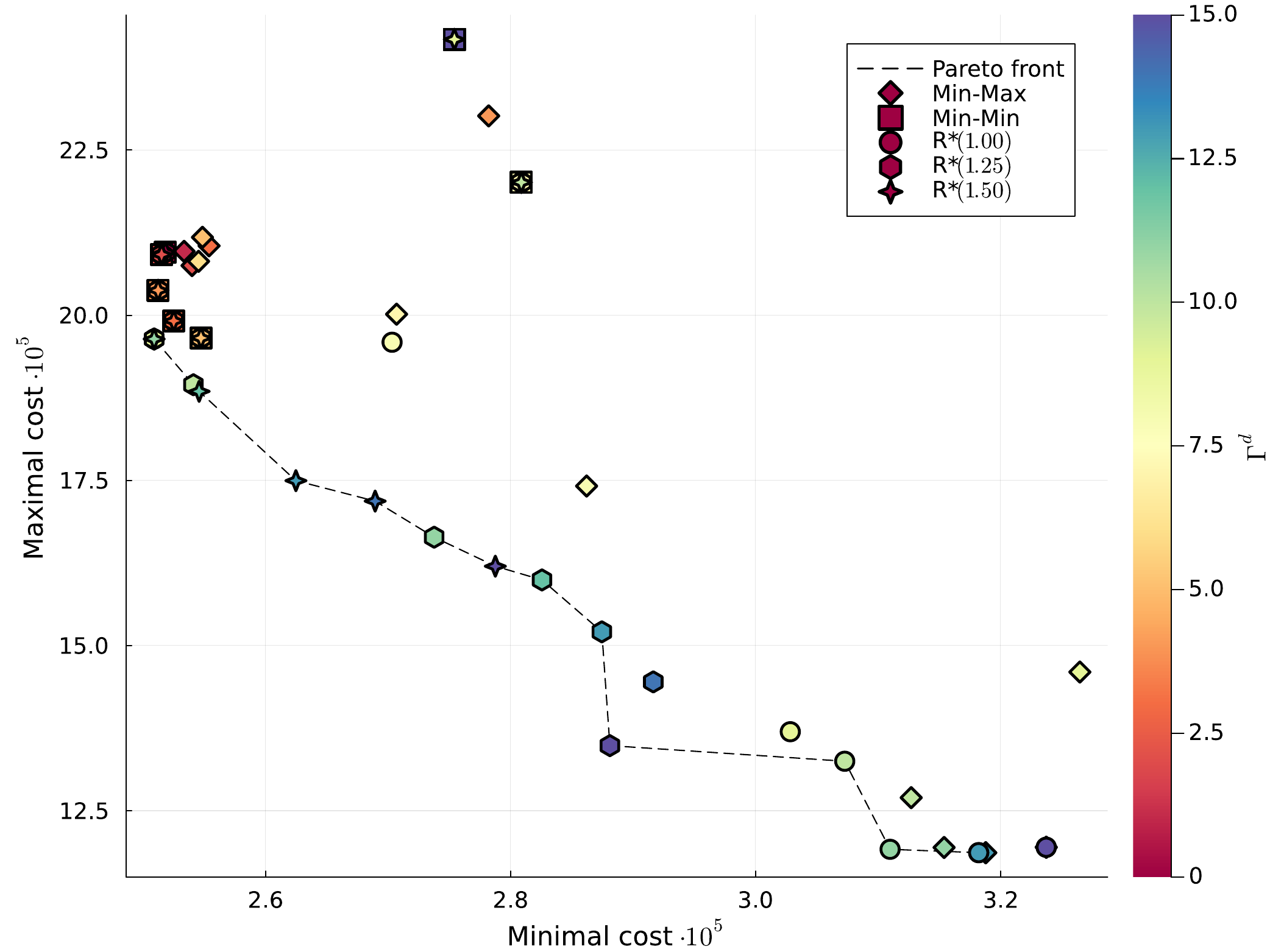}
	\caption{The minimal and maximal  costs of production plans for the instances with $T=15$ and $T^+=31$. The dashed line represents the Pareto front. The colors represent the budget~$\Gamma^d$.}
	\label{fig:minmax_pareto_15_31}
\end{figure}

The solutions to the \textsc{Min-Min} problem, represented as squares, are clustered in the top-left quadrant. They have small minimum, but large maximum cost. Note that they are located on the Pareto front. Therefore, they should be taken into account by optimistic decision-makers. On the other hand,  the solutions to the  \textsc{Min-Max} problem, represented as rotated squares, explore primarily the bottom-right quadrant of the figure. Interestingly all of them, except for the one with $\Gamma^d=10$, are beyond the Pareto front. Hence, there are other solutions that dominated them if we only look at the extreme values for  the costs. The solutions to the \textsc{R*}  problem clearly expand the Pareto front. They are not achievable by solving both the  \textsc{Min-Min} and  \textsc{Min-Max} problems. Hence, the   \textsc{R*} approach significantly expands the number of candidate production plans. 
The ultimate choice of production plan depends on other factors such as decision-maker attitude towards-risk or available budget. Observe also that for the instance with $T=15$ and $T^+=31$, the nominal production plan is not on the Pareto front (for example, the  \textsc{R*}$(1.00)$-production plan with $\Gamma^d=5$ has smaller both the minimum and the maximum cost).

\section{Conclusion}

In this paper we have discussed a single-item  lot sizing problem with uncertain lead-times. We used a simple model of uncertainty in which the lead-times can take some values within specified ranges. We have imposed two assumptions on possible lead-time scenarios. Firstly, all planning periods have the same nominal common lead-time. Secondly, only lead-time scenarios without crossing of orders are allowed. Both assumptions are realistic in practical applications. Using them, we can construct polynomial-time algorithms for computing the best (optimistic) and the worst (pessimistic) lead-time scenario for a given production plan. We can also solve the \textsc{Min-Max} problem in polynomial time (if the setup costs are 0) and the \textsc{Min-Min} problem in pseudopolynomial time. For the more general \textsc{R*} problem, a compact MIP formulation can be constructed. The main motivation for using the \textsc{R*} criterion is in extending the set of candidate production plans. By performing some experiments, we showed that solving
the \textsc{R*} problem can give us solutions that are not achievable by the traditional robust min-max approach.

There are several interesting open problems regarding the complexity of the  model considered. For example, we can solve the \textsc{Min-Min} problem in pseudopolynomial time, but the complexity of this problem when the setup costs are 0 is open (we do not know if its is weakly NP-hard in this case).  For the \textsc{Min-Max} problem with setup-costs a compact MIP formulation exists. The problem is NP-hard, but it is not known whether it can be solved in pseudopolynomial time. It could  also be interesting to extend the problem to many suppliers, like in~\cite{TON22}.

\subsubsection*{Acknowledgements}
Adam Kasperski, Szymon Wróbel and Paweł Zieliński were supported by
 the National Science Center, Poland, grant 2022/45/B/HS4/00355.


\appendix

\section{Appendix}
\label{dod}

\begin{proof}[The proof of Theorem~\ref{thmcompl}]

We begin by recalling the \textsc{Balanced Partition} problem, which is NP-complete~\cite{GJ79}.
 Let $\pmb{a}=(a_1,a_2,\dots,a_n)$ be a collection of nonnegative integers with even $n\geq 2$ and $S=\frac{1}{2}\sum_{i\in [n]} a_i$. We ask if there is a subset of indices $I\subseteq [n]$ of size $|I|= n/2$, such that $\sum_{i\in I} a_i =S$. 

We show a reduction from an instance of \textsc{Balanced Partition} problem to an instance of the problem~\textsc{R*}.
Let $n\geq 2$ and $\pmb{a}=(a_1,a_2,\dots,a_n)$ be an instance of \textsc{Balanced Partition}. Assume w.l.o.g. that $a_1\leq a_2\leq \dots \leq a_n$. Let us construct the following instance of the problem~\textsc{R*} (see~(\ref{probr}))  under the lead-time uncertainty set~$\mathcal{L}$:
\begin{itemize}
	\item $T=n$, $T^+=2n+1$, so there are $n$ planning periods and $n+1$ future periods,
	\item $c^P_t=0$, $C_t=1$, $c^S_t=a_t$ for $t\in [n]$,
	\item $c^I_{n+i}=2a_{i+1}-2a_{i}$ for $i\in [n-1]$, $c^I_{2n}=2S-2a_n$, $c^I_{2n+1}=1$,
	\item $\widehat{l}=n$, $l_t\in \{n\dots,2n+1-t\}$, $t\in [n]$, lead-dime uncertainty set is $\mathcal{L}$,
	\item $d_t=0$ for $t=n+1,\dots, 2n$ and $d_{2n+1}=\frac{n}{2}$,
	\item $B=(n-1)S$.
\end{itemize}
An example of the reduction  is shown in Figure~\ref{figcompl}.
\begin{figure}[ht]
	\centering
	\includegraphics[height=6cm]{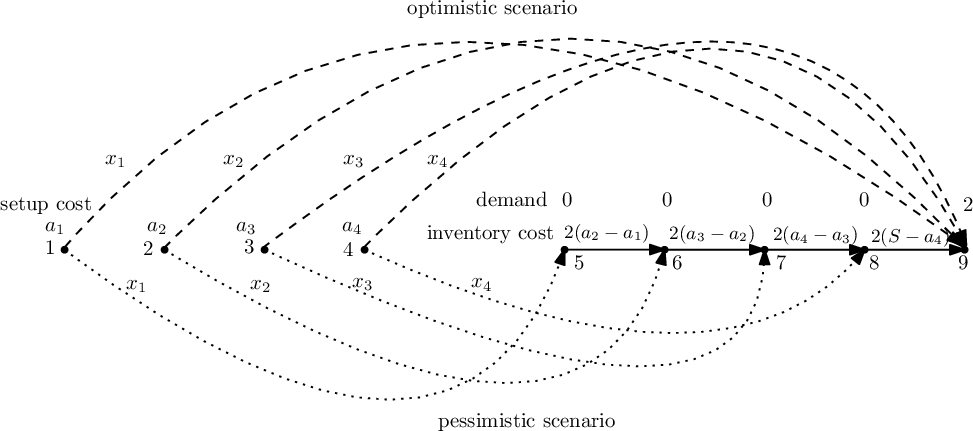}
	\caption{A sample reduction for $\pmb{a}=(a_1,a_2,a_3,a_4)$, $c^P_t=0$ for $t=1,\dots,4$ and $x_t\in[0,1]$ for $t=1,\dots,4$.
	The dashed arcs correspond to the optimistic lead-time scenario $\pmb{\ell}^o=(8,7,6,5)$ and the dotted arcs correspond to the pessimistic lead-time scenario $\pmb{\ell}^p=(4,4,4,4)$.  } \label{figcompl}
\end{figure}

Observe that, in the constructed instance, there are no backordering costs and  there is a constant capacity equal to~1 in every period $t\in [T]$. The cumulative demand $D_t=0$ for $t=n+1,\dots,2n$ and $D_{2n+1}=\frac{n}{2}$. Also, in any optimal solution $I_{2n+1}=0$.
Given a lead-time scenario $\pmb{\ell}=(l_1,\dots,l_n)\in \mathcal{L}$, we can rewrite~(\ref{spp1})-(\ref{spp6}) as
\begin{align}
 \min &\sum_{t\in [n]} c_t^S y_t+\sum_{t=n+1}^{2n} c^{I}_t I_{t}&\label{spp1b}\\ 
\text{ s.t. }& \sum_{i\in\{i\in [n]: i+l_i\leq t\}} x_{i}=I_t & t=n+1,\dots,2n \label{spp2b}\\
& \sum_{t\in [n]} x_{t}=\frac{n}{2} \\
& 0\leq x_{t}\leq y_t & t\in [n] \label{spp4b}\\
        &I_{t}\geq 0  & t=n+1,\dots,2n+1\label{spp5b}\\
        &y_{t}\in \{0,1\}  & t\in [n].\label{spp6b}
 \end{align}       
The above model can be further simplified by substituting $I_t$, given by~(\ref{spp2b}) into~(\ref{spp1b}).
It is easy to see that $\sum_{t\in [n]} y_t\geq \frac{n}{2}$, since $C_t=1$ for each $t\in [n]$. 
Clearly,  the best lead-time scenario for any production plan $\pmb{x}$ is $\pmb{\ell}^o=(2n, 2n-1,\dots,n)$, that is when the production for each period $t\in [n]$ arrives at the last period $2n+1$. 
Note that no inventory cost is then incurred, and the total cost of $\pmb{x}$ is simply its setup cost.
$$\sum_{t\in [n]} c^S_t y_t=\sum_{t\in [n]} a_t y_t.$$
On the other hand, the worst scenario for any $\pmb{x}$ is $\pmb{\ell}^p=(n,n,\dots,n)$, that is when  production arrives at the earliest possible future period. Observe that all lead times are equal to the nominal lead time $n$. The inventory cost of $\pmb{x}$ is then
\begin{align*}
\sum_{t=n+1}^{2n} c^{I}_t \sum_{i\in\{i\in [n]: i+n\leq t\}} x_{i}&=\sum_{t=n+1}^{2n} c^{I}_t \sum_{i=1}^{t-n} x_i=\sum_{t=1}^{n} c^{I}_{n+t} \sum_{i=1}^{t} x_i\\
&=\sum_{t=1}^{n}x_t\sum_{i=t}^{n} c^{I}_{n+i}=\sum_{t=1}^{n}x_t\left(\sum_{i=t}^{n-1} (2a_{i+1}-2a_{i})+2(S-a_n)\right)\\
& =\sum_{t\in [n]} 2(S-a_t)x_t.
\end{align*}
Therefore, the overall setup and inventory cost under the pessimistic scenario becomes
$$\sum_{t\in [n]} a_ty_t+\sum_{t\in [n]} 2(S-a_t)x_t.$$
Hence, a model of the problem~\textsc{R*} (see~(\ref{probr})), for the constructed  instance,  takes the following form:
$$
\begin{array}{llll}
	\min & \displaystyle \sum_{t\in [n]} a_t y_t \\
		\text{s.t.}& \displaystyle\sum_{t\in [n]} a_t y_t+\sum_{t\in [n]} 2(S-a_t)x_t\leq B\\
		&\displaystyle \sum_{t\in [n]} x_t= \frac{n}{2}\\
		& 0\leq x_t \leq y_t & t\in [n]\\
		& y_t\in \{0,1\} & t\in [n].
\end{array}
$$
Substituting $B=(n-1)S$ yields the following equivalent reformulation:
\begin{align}
	\min & \displaystyle \sum_{t\in [n]} a_t y_t \label{red0}\\
	\text{s.t. }\;	& \displaystyle 2\sum_{t\in [n]} a_t x_t-\sum_{t\in [n]} a_ty_t \geq S \label{red1}\\
		&\displaystyle \sum_{t\in [n]} x_t= \frac{n}{2} \label{red001}\\
		& 0\leq x_t \leq y_t & t\in [n] \label{red002}\\
		& y_t\in \{0,1\} & t\in [n]. \label{red2}
\end{align}

The objective value is at least $S$ for every feasible solution $(\pmb{x}, \pmb{y})$ to~(\ref{red0})-(\ref{red2}). Indeed, suppose that there exists a solution $(\pmb{x}, \pmb{y})$ whose objective value is smaller than $S$. Then $\sum_{t \in [n]} a_t x_t \leq \sum_{t \in [n]} a_t y_t$ and $2 \sum_{t \in [n]} a_t x_t - \sum_{t \in [n]} a_t y_t \leq \sum_{t \in [n]} a_t y_t < S$, which implies that the cost-budget constraint~(\ref{red1}) is violated. Also, in any optimal solution there is exactly $n/2$ positive variables $y_i$, $i\in [n]$. Indeed the number of positive $y_i$ cannot be less than $n/2$ due to~(\ref{red001}) and~(\ref{red002}). To see that it also cannot be greater than $n/2$, consider a  feasible  $\pmb{y}$. We can now find a feasible $\pmb{x}$ by assigning $x_i=1$ to $n/2$ variables $x_i$ with largest $a_i$ such that $y_i=1$. If the number of positive $y_i$ is greater than $n/2$, then we can decrease the objective value by assigning $y_i=0$ to some of them (recall that all numbers $a_i$ are positive).

It remains to show that the answer to an instance of \textsc{Balanced Partition} is ``yes'' if and only if the optimal objective value in model~(\ref{red0})–(\ref{red2}) (i.e., the model of problem~\textsc{R*}, see~(\ref{probr}), for the constructed instance) is equal to $S$. Suppose that the answer to an instance of \textsc{Balanced Partition} is ``yes''. Consider a solution in which $y_i = x_i=1$ if $i \in I$, $|I|= n/2$. 
 This solution is feasible for~(\ref{red0})-(\ref{red2}), and its objective value equals $S$, which is therefore optimal.
 Consider an optimal solution to~(\ref{red0})-(\ref{red2}) with the objective value of $S$. The number of positive $y_i$ in this solution is $n/2$. Therefore, the answer to \textsc{Balanced Partition} is "yes" due to~(\ref{red0}) and~(\ref{red2}). 
 Thus, the problem~\textsc{R*} is NP-hard. We have proved the first part of the theorem.
  
 To prove the second part of the theorem, it suffices to modify the reduction by setting $c_t^P = a_t$ and $c_t^S = 0$ for $t \in [n]$. The remaining part of the reduction remains unchanged. Then $x_t \in \{0,1\}$ for each $t \in [T]$. The proof is then almost identical to the previous one.
\end{proof}

\begin{proof}[The proof of Theorem~\ref{thmminmin}]
We will use a reduction from the following \textsc{Partition} problem, which is NP-complete~\cite{GJ79}. Let $\pmb{a}=(a_1,a_2,\dots,a_n)$ be a collection of nonnegative integers and $S=\frac{1}{2}\sum_{i\in [n]} a_i$. We ask if there is a subset of indices $I\subseteq [n]$, such that $\sum_{i\in I} a_i =S$.
Let $n$ and 
$\pmb{a}=(a_1, a_2, \dots, a_n)$ be an instance of \textsc{Partition}. 
The construction of an instance of the problem~(\ref{wleadprob})
under the lead-time uncertainty set~$\mathcal{L}$, with the no order 
crossover assumption~(\ref{clun2}) relaxed, is as follows:
\begin{itemize}
	\item $T=n$, $T^+=2n+1$, so there are $n$ planning periods and $n+1$ future periods,
	\item $x_t=a_t$ for $t\in [n]$,
	\item $c^I_{n+i}=0$ for $i\in [n-1]$, $c^I_{2n}=M$, $c^I_{2n+1}=0$,
	\item $\widehat{l}=n$, $l_t\in \{n\dots,2n+1-t\}$, $t\in [n]$, 
	\item $d_t=0$ for $t=n+1,\dots, 2n-1$,  $d_{2n}=S$  and $d_{2n+1}=S$,
\end{itemize}
where $M=2S$. A sample reduction is shown in Figure~\ref{figcompl1}.

\begin{figure}[ht]
	\centering
	\includegraphics[height=5cm]{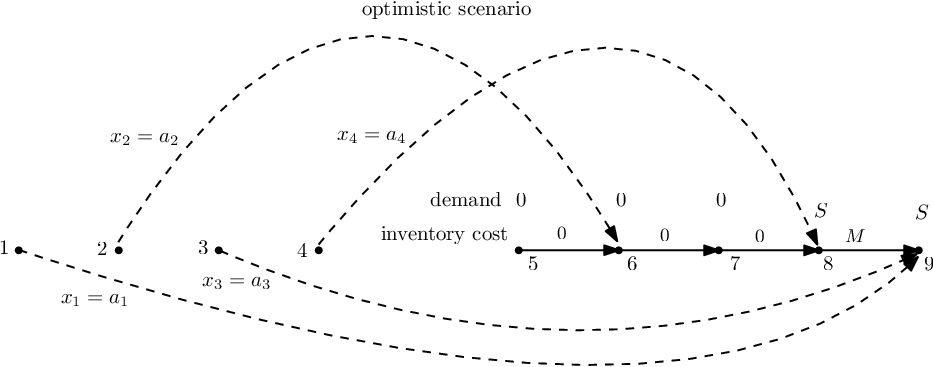}
	\caption{A sample reduction for $\pmb{a}=(a_1,a_2,a_3,a_4)$.
	The dashed arcs correspond to the optimistic lead-time scenario $\pmb{\ell}^o=(8,4,6,4)$.} \label{figcompl1}
\end{figure}

We now prove that the answer to an instance of \textsc{Partition} is ``yes'' if and only if 
the optimal objective value in  the problem~(\ref{wleadprob}), for the constructed instance, is equal to~$0$.

Assume that the answer to an instance of  \textsc{Partition} is ``yes'', i.e.
there is $I\subseteq [n]$, such that $\sum_{t\in I} a_t =S$. 
For the production plan $x_t=a_t$ for $t\in [n]$,
 we form an 
optimistic lead-time scenario~$\pmb{\ell}^{o}$ as follows: 
$l^{o}_t=n$ if $t \in I$ and  production~$a_t$ arrives at period~$n+t$
to meet demand $d_{2n}$ with zero inventory holding cost; 
 $l^{o}_t=2n+1-t$ if $t \not\in I$ and 
production~$a_t$ directly arrives at period~$2n+1$
 to meet demand $d_{2n+1}$ and no inventory holding costs are incurred. 
 Clearly, demands $d_{2n}$ and $d_{2n+1}$ are 
met  exactly at level~$S$.
Thus, the optimal objective value of problem~(\ref{wleadprob}) is 
equal to~$0$.

Suppose that $\pmb{\ell}^{o}$ is an optimistic lead-time scenario for 
problem~(\ref{wleadprob}), for which the  objective value is~$0$ 
under the production plan~$\pmb{x}$ such that $\sum_{t\in [n]}x_t = 2S$. 
This implies that zero inventory holding costs are incurred for~$\pmb{x}$. 
Consequently, no inventory is carried  from periods~$2n$ to~$2n+1$ at 
the high cost~$M$ to meet the demands $d_{2n}$ and $d_{2n+1}$ at level~$S$. 
This means that no single production lot at any period~$t \in [T]$ is 
used to satisfy both $d_{2n}$ and $d_{2n+1}$ simultaneously. 
Let $I=\{t \in [n] : l^{o}_t = 2n+1-t\}$ be the set of periods where 
production arrives directly at period~$2n+1$ to meet demand~$d_{2n+1}=S$. 
Because$\sum_{t \in I} a_t = S$, the answer to 
\textsc{Partition} is ``yes''.

The inapproximability follows from the fact that any polynomial-time 
approximation algorithm for problem~(\ref{wleadprob}), when applied 
to the instance constructed above, could decide in polynomial time 
whether the objective value is zero, thereby determining whether the 
answer to the  corresponding \textsc{Partition} instance is ``yes''.
\end{proof}

\end{document}